\documentclass[journal]{IEEEtran}

\usepackage{epsfig,amsmath,amssymb,epsf,amsthm,scalefnt,multirow,subfig}
\usepackage{float}
\usepackage{cite}
\usepackage{psfrag}
\usepackage[dvipsnames]{xcolor}
\newtheorem{theorem}{Theorem}
\newtheorem{lemma}{Lemma}
\newtheorem*{lemma*}{Lemma}

\newtheorem{corollary}{Corollary}

\newtheorem{proposition}{Proposition}

\def\b0{{\pmb{0}}} 

\def\romanone{\uppercase\expandafter{\romannumeral 1}}
\def\romantwo{\uppercase\expandafter{\romannumeral 2}}
\def\romanthree{\uppercase\expandafter{\romannumeral 3}}
\def\romanfour{\uppercase\expandafter{\romannumeral 4}}
\def\romanfive{\uppercase\expandafter{\romannumeral 5}}

\theoremstyle{remark}

\def\black{\textcolor{black}}

\begin{document}

% paper title
%\title{A Tractable Approach to Analyzing Satellite Cluster Through Stochastic Geometry}
\title{Analyzing Downlink Coverage in Clustered Low Earth Orbit Satellite Constellations: A Stochastic Geometry Approach}

%author names and IEEE memberships
\author{\IEEEauthorblockN{Miyeon Lee, \textit{Graduate Student Member, IEEE}, 
 Sucheol Kim \textit{Member, IEEE}, Minje Kim \textit{Graduate Student Member, IEEE}, Dong-Hyun Jung, \textit{Member}, and Junil Choi, \textit{Senior Member, IEEE}}\\
\thanks{This work was partly supported by Institute of Information\& communications Technology Planning \& Evaluation(IITP) grant funded by the Korea government(MSIT)(No.2021-0-00847, Development of 3D Spatial Satellite Communications Technology), Korea Research Institute for defense Technology planning and advancement(KRIT) grant funded by the Korea government(DAPA(Defense Acquisition Program Administration))(KRIT-CT-22-040, Heterogeneous Satellite constellation based ISR Research Center, 2022), and the Institute of Information \& Communications Technology Planning \& Evaluation(IITP)-ITRC(Information Technology Research Center) grant funded by the Korea government(MSIT)(IITP-2025-RS-2020-II201787, contribution rate 20\%).}
%A conference version of this paper was presented in part at the IEEE PIMRC, 2024 [DOI: 10.1109/PIMRC59610.2024.10817287].
%This work was supported by Korea Research Institute for defense Technology planning and advancement(KRIT) grant funded by the Korea government(DAPA(Defense Acquisition Program Administration)) (KRIT-CT-22-040 , Heterogeneous Satellite constellation based ISR Research Center, 2024)
\thanks{Miyeon Lee , Minje Kim, and Junil Choi are with the School of Electrical Engineering, Korea Advanced Institute of Science and Technology, Daejeon, 34141, South Korea (e-mail: \{mylee0031; mjkim97; junil\}@kaist.ac.kr).}
\thanks{Sucheol Kim is with the Satellite Communication Research
Division, Electronics and Telecommunications Research Institute, Daejeon, 34129, South Korea (e-mail: loehcusmik@etri.re.kr).}
\thanks{Dong-Hyun Jung is with the School of Electronic Engineering, Soongsil University, Seoul, 06978, South Korea (e-mail: dhjung@ssu.ac.kr).}}

\maketitle
%\squeezeup\squeezeup\squeezeup\squeezeup
\begin{abstract}
Satellite networks are emerging as vital solutions for global connectivity beyond 5G. As companies such as SpaceX, OneWeb, and Amazon are poised to launch a large number of satellites in low Earth orbit, the heightened inter-satellite interference caused by mega-constellations has become a significant concern. To address this challenge, recent works have introduced the concept of satellite cluster networks where multiple satellites in a cluster collaborate to enhance the network performance. In order to investigate the performance of these networks, we propose mathematical analyses by modeling the locations of satellites and users using Poisson point processes, building on the success of stochastic geometry-based analyses for satellite networks. In particular, we suggest the lower and upper bounds of the coverage probability as functions of the system parameters, including satellite density, satellite altitude, satellite cluster area, path loss exponent, and the Nakagami parameter $m$. We validate the analytical expressions by comparing them with simulation results. Our analyses can be used to design reliable satellite cluster networks by effectively estimating the impact of system parameters on the coverage performance.
\end{abstract}


\section{Introduction}\label{sec1}
Satellite communications have garnered interest as a promising solution to achieve ubiquitous connectivity. According to the International Telecommunication Union (ITU) report \cite{ITU}, nearly half of the global population remains unconnected. Satellites are expected to address this connectivity gap, offering a viable solution that alleviates the cost burden faced by telecom operators when deploying terrestrial base stations \cite{Chen:2021}. Traditionally, geostationary orbit (GEO) satellites were the primary choice for satellite networks due to the extensive coverage facilitated by their high altitudes. However, the extremely high altitude of GEO satellites, \black{approximately 35,800~km}, presents challenges such as high latency and low area spectral efficiency, making them less suitable for 5G applications. The industry has shifted {toward} adopting low Earth orbit (LEO) satellites, \black{defined by orbits at altitudes up to 2,000~km}, to reduce latency, increase service density, and facilitate cost-effective launches \cite{Liu:2021, Hassan:2020, Wang:2022}. Companies like SpaceX, OneWeb, and Amazon are planning to deploy large constellations in LEOs, often referred to as \black{mega-constellations} \cite{Del:2019}. \black{However, in mega-constellations, increasing the number of satellites can lead to performance degradation due to severe inter-satellite interference. To address this challenge, satellite cooperation strategies such as LEO remote sensing \cite{Jian:2020} and satellite clusters \cite{Jung:2023},\cite{Homssi:2022} have been developed to manage dense satellite networks and mitigate inter-satellite interference effectively.}

As mega-constellations complicate simulation-based network performance analyses, the need for new tools arises for effective assessment. Stochastic geometry is introduced to mathematically analyze performance, modeling the random behavior of nodes in wireless networks, including base stations and users through point processes such as binomial point processes (BPPs) or Poisson point processes (PPPs) \cite{Haenggi:2012}. Previous studies \cite{Haenggi:2009,Andrews:2011,Dhilon:2012,Novlan:2013,Tan:2014,Afshang:2017} have conducted stochastic geometry-based analyses for wireless networks using various performance metrics for terrestrial networks. A framework for modeling wireless networks was proposed in \cite{Haenggi:2009}, addressing the outage probability, network throughput, and capacity. \black{In \cite{Andrews:2011}, {the authors derived the coverage probability and average achievable rate using a PPP and justified the characterization of the distribution of base stations as the PPP.} In \cite{Dhilon:2012}, the coverage probability and average achievable rate were characterized under downlink heterogeneous cellular networks modeled using the PPP.} In \cite{Novlan:2013}, the coverage probability for uplink cellular networks modeled by the PPP was studied. The authors in \cite{Tan:2014} derived the downlink coverage probability considering base station cooperation within the PPP framework. A BPP model for cache-enabled networks was used to characterize the downlink coverage probability and network spectral efficiency in \cite{Afshang:2017}.

Building on stochastic geometry-based analyses in terrestrial networks, recent works have extended its application to satellite networks. Given the envisaged prevalence of large-scale satellite constellations, it is crucial to analyze the network performance in scenarios where multiple satellites serve a user, moving beyond a single-satellite setup. In this work, we focus on satellite cluster networks, a pertinent solution in mega-constellations, and employ stochastic geometry to characterize the coverage probability, providing an effective means of estimating the impact of system parameters.
%Given the potential for large-scale satellite constellations in future networks, analyzing the performance of the cluster network is crucial, as it can help mitigate inter-satellite interference in mega-satellite constellations.

\subsection{Related Works}
\black{Recent works \cite{Okati:2020, Anna:2021, Park:2023, Al:2021, Al2:2021, Okati:2022} {have utilized stochastic geometry to evaluate the downlink coverage performance of satellite networks.} {In \cite{Okati:2020}, the downlink coverage probability and average achievable rate were derived for} a scenario in which users are associated with the nearest satellite, with satellite locations modeled using a BPP on a sphere, assuming that the satellite distribution in a certain area follows the binomial distribution. Coverage performance in another scenario, where satellite gateways act as relays between users and LEO satellites, was investigated in \cite{Anna:2021} using the BPP for satellite location distribution modeling. In \cite{Park:2023}, the downlink coverage probability for a scenario where users are associated with the nearest satellite was derived using a PPP to model satellite location distribution, assuming that satellite distribution in a certain area follows the Poisson distribution. Additionally, the authors demonstrated that the PPP could accurately represent the actual Starlink constellation in terms of visible satellite numbers. The authors in \cite{Al:2021} focused on deriving coverage probability from the contact angle distribution, incorporating path loss modeling and line-of-sight probability formulations while modeling satellite locations using the PPP. In \cite{Al2:2021}, {satellite altitude was optimized to maximize} downlink coverage probability, with satellite locations modeled using the PPP. Additionally, to address variations in satellite densities across latitudes, coverage probabilities and average achievable data rates were derived using a non-homogeneous PPP in \cite{Okati:2022}. Based on these works, it is confirmed that stochastic geometry provides a robust analytical framework for satellite network performance analyses. {However, these performance analyses have primarily focused on scenarios in which a single satellite serves a user, without incorporating cooperative strategies among satellites to manage interference effectively in large satellite constellations.}}

{The proliferation of LEO satellites underscores the necessity for a theoretical understanding of system parameters to design networks effectively, with attention to interference mitigation, as increasing network density can reduce coverage probability due to rising interference. While performance analyses for single satellite-to-user scenarios are well-established, studies on cooperative transmission involving multiple satellites through clustering are incipient. In early work \cite{Kim:2024}, satellite cooperative transmission using coordinated beamforming with up to $\mathrm{K}$-nearest satellites was explored, and the coverage probability was derived. In contrast, we propose a fundamentally distinct approach by defining a cluster area and adopting non-coherent joint transmission for cooperative transmission. This technique enables data transmission with less stringent requirements for synchronization and channel state information compared to other techniques \cite{Lee:2012}, leveraging satellite clusters formed based on their geographical proximity to the user.}

%The proliferation of satellites underscores the necessity for a theoretical understanding of system parameters to design networks effectively. In light of this, our stochastic geometry-based analyses shift the focus from scenarios where a single satellite serves to those involving multiple satellites, examining the impact of system parameters on network performance. \black{We incorporate satellite clustering as a mechanism for cooperation, enabling multiple satellites within a cluster area to serve a user through non-coherent joint transmission. \blue{This method allows for the joint transmission of data} to a user with less stringent requirements for synchronization and channel state information compared to other techniques \cite{Lee:2012}. }
%It is worth noting that while cooperative transmission has previously been investigated in terrestrial networks \cite{Tan:2014}, this study redirects attention from terrestrial networks to satellite networks. 

\subsection{Contributions}
We focus on \black{LEO satellite networks}, where a user is served by multiple satellites in a cluster area, while satellites outside this area act as interfering nodes. To evaluate network performance, we employ stochastic geometry to model the locations of satellites and users. In particular, the satellite locations follow a homogeneous PPP on the surface of a sphere defined by the satellite orbit radius, and the user locations also follow a homogeneous PPP on the surface of a sphere defined by the Earth radius.

Distinct from the $\mathsf{K}$-nearest clustering approach in \cite{Kim:2024}, our derivation considers the satellite cluster area for cooperative transmission, ensuring proximity both between satellites and users as well as among satellites. In line with existing studies and the growing focus on satellite cooperation, this work broadens the scope of mathematical performance analysis for satellite cluster networks.

The main contributions are summarized as follows:
\begin{itemize}

\item Using stochastic geometry, we derive the lower and upper bounds of the coverage probability {for a user served by multiple cooperating satellites.} To address the effects of rounding operations during bound derivation, we propose a heuristic equation that incorporates relative distances. The results demonstrate that the derived bounds effectively encapsulate the simulation outcomes, with the heuristic approach also exhibiting reasonable {alignment}. {Furthermore, for the special case where small-scale fading follows the Rayleigh distribution,} we deduce \black{{compact expressions}.} 

\item {As a preliminary step in deriving bounds of the coverage probability, we address compound terms, i.e., the accumulated cluster power and interference power, induced by joint transmissions via cooperative satellites to mitigate inter-satellite interference. By leveraging properties that the sum over Poisson processes on a finite sphere surface can be approximated as a Gamma random variable, we determine the corresponding shape and scale parameters, enabling simultaneous management of the accumulated cluster power and interference power.}

%\red{To handle the sum over Poisson processes simultaneously during deriving the bounds of the coverage probability, we manage compound terms, i.e., the accumulated cluster power and interference power, induced by joint transmissions via cooperative satellites to mitigate inter-satellite interference. We determine the relevant shape and scale parameters and show the validity in satellite networks.}
%However, a challenge arises since these terms,the accumulated cluster power and interference power,originate from the same source of randomness. To address this, we leverage properties where sums over Poisson processes on a finite sphere surface can be approximated as a single random variable: the Gamma random variable
%and demonstrate that these approximations closely align with empirical results for sums over Poisson processes.
%In the cooperative satellite network, a user may be served by multiple satellites in a cluster area within a finite space.

\item {\black{From the main results, we interpret the following observations: (i) the coverage probability obtained through satellite cooperation increases with the cluster area, and (ii) the impact of increasing the line-of-sight (LOS) on coverage probability depends on the cluster area.} We perform detailed analyses on how system parameters impact the coverage probability in satellite cluster networks, offering essential insights for system design.}
\end{itemize}

\subsection{Organizations and Notations}
The remainder of the paper is organized as follows. We describe the system model and the performance metric framework in Section~\ref{sec2}. Baseline preliminaries and the derived results of the lower and upper bounds of the coverage probability are presented in Section~\ref{sec3}. In Section~\ref{sec4}, simulation results validate the analytical expressions and analyze the impact of system parameters on the coverage probability. Finally, conclusions follow in Section~\ref{sec5}.
\subsubsection*{\black{Notations}} Sans-serif letters ($\mathsf{X}$) denote random variables while serif letters ($x$) denote their realizations or scalar variables. Lower boldface symbols ($\mathbf{x}$) denote column vectors. $\mathbb{E}[ \mathsf{X} ]$ and $\text{Var}[\mathsf{X}]$ denote the expectation and variance. $\lVert \mathbf{x} \rVert$ denotes the Euclidean norm. $f_\mathsf{X}(\cdot)$ and $F_\mathsf{X}(\cdot)$ are used to denote the probability density function (PDF) and the cumulative density function (CDF). The Laplace transform of $\mathsf X$ is defined as $\mathcal{L}_{\{\mathsf X\}}(s)=\mathbb{E}\left[e^{-s\mathsf X}\right]$. $P(\cdot)$ denotes probability. $\Gamma(k,\theta)$ is used to denote the Gamma distribution with the shape parameter $k$ and scale parameter $\theta$. The Gamma function is defined as ${\Gamma(z) = \int_{0}^{\infty} t^{z-1}e^{-t}dt }$. Specifically, for positive integer $z$, the Gamma function is defined as $\Gamma(z)=(z-1)!$. ${}_2F_1\left(\cdot, \cdot ; \cdot ; \cdot\right)$ is the Gauss's hypergeometric function, {and ${}_1F_1\left(\cdot; \cdot ; \cdot \right)$ is the confluent hypergeometric function.} $\lceil \cdot \rceil$ and $\lfloor \cdot \rfloor$ denote the round-up and -down operations. 

\section{System Model}\label{sec2}
We consider a downlink satellite network depicted in Fig.~\ref{fig_system}. Satellites are located on a sphere with radius $R_{\mathrm  S}$ according to a homogeneous spherical Poisson point process (SPPP) $\Pi = \{\mathbf{x}_1, \ldots, \mathbf{x}_{\mathsf N_{\mathrm S}}\}$ where $\mathbf{x}_{l}$, $l \in \{1, \ldots, \mathsf N_{\mathrm S}\}$, is the position of the satellite $l$ in a spherical coordinate system with {polar} and azimuth angles spanning from $0$ to $\pi$ and $0$ to~$2\pi$, respectively. The total number of satellites, ${\mathsf N_{\mathrm S}}$, follows the Poisson distribution with density $\lambda_{\mathrm S}$. {While the homogeneous PPP does not precisely capture latitudinal density variations in satellite distributions, it buttresses the approximation of actual distributions observed in Starlink satellite constellations, as demonstrated in \cite{Park:2023}.} Single-antenna users, positioned on the Earth's surface with radius $R_{\mathrm  E}$, are also distributed according to a homogeneous SPPP, denoted as $\Pi_\mathrm{U} = \{\mathbf{u}_1, \ldots, \mathbf{u}_{\mathsf N_{\mathrm U}}\}$, where ${\mathsf N}_{\mathrm U}$ is the total number of users on the Earth, which follows the Poisson distribution with density $\lambda_\mathrm U$. In our downlink performance analyses, we focus on a typical user located at $\mathbf{u}_1=(0,0,R_{\mathrm E})$, which exemplifies spatially averaged performance without compromising generality. This approach is supported by the Slivnyak's theorem, which applies to a homogeneous PPP \cite{Haenggi:2012}.
%According to the Slivnyak's theorem \cite{Haenggi:2012}, we focus on a typical user located at $\mathbf{u}_1=(0,0,R_{\mathrm  E})$ for our downlink performance analyses without loss of generality.
%represents a point in a spherical coordinate system with elevation and azimuth angles spanning from $0$ to $\pi$ and $0$ to $2\pi$, respectively%
%\subsection{Spatial Distributions of Satellites and Users}
\subsection{Observable Spherical Dome and Cluster Area}
Given that satellites below a specified elevation angle may be invisible to the typical user, we introduce the elevation angle $\theta_{\mathrm{min}}$ to define the observable spherical dome $\mathcal{A}$. To consider the visibility constraint, we define the maximum distance to the satellites in the observable spherical dome as $R_{\mathrm{max}}$, which is given by
\begin{align} \label{cos}
R_{\mathrm{max}}= -R_{\mathrm  E} \sin \theta _{\mathrm{ min}}+\sqrt {R_{\mathrm  S}^{2} -R_{\mathrm  E}^2\cos^2 \theta _{\mathrm{ min}}},
\end{align}
using the law of cosines. Then we calculate the surface area of the observable spherical dome, $|\mathcal{A}|$, using the Archimedes' Hat box theorem as \cite{Cundy:1989}
\begin{align}\label{dome}
|\mathcal {A}|=&2\pi R_{\mathrm  S}\left ({R_{\mathrm  S} - R_{\mathrm  E} - R_{\mathrm{ max}} \cos \left(\pi /2 - \theta _{\mathrm{min}}\right)}\right).
\end{align}
To define the satellite cluster area $\mathcal{A}_{\mathrm{clu}}$ where the satellites in the region cooperatively serve a user, we introduce the polar angle $\phi_{\mathrm {clu}}$, which represents the maximum zenith angle of the multiple satellites in the cluster area as shown in Fig.~\ref{fig_system}. The maximum distance from the typical user to the cluster area, $R_{\mathrm{clu}}$, and the surface area of the cluster $|\mathcal{A}_{\mathrm{clu}}|$ are then respectively given by
\begin{align}
R_{\mathrm{clu}}=\sqrt{R_{\mathrm  S}^2+R_{\mathrm  E}^2-2R_{\mathrm{S}}R_{\mathrm{E}}\cos\phi_{\mathrm{clu}}},
\end{align}
\begin{align}
|\mathcal {A}_{\mathrm{clu}}|=&2\pi R_{\mathrm  S}\left(R_{\mathrm  S}-R_{\mathrm  E} -\frac{R_{\mathrm  S}^2-R_{\mathrm  E}^2-R_{\mathrm{clu}}^2}{2 R_\mathrm  E}\right).
\end{align}
{Additionally, we define the minimum distance to the cluster area as $R_{\mathrm{min}}=R_\mathrm  S -R_\mathrm  E$.}

\begin{figure} [t!]
	\centering
	\includegraphics[width=0.81\columnwidth]{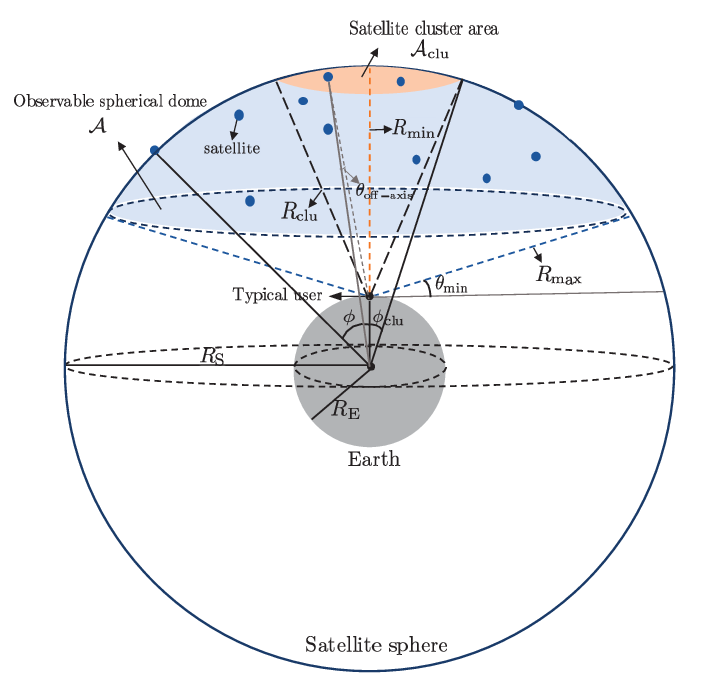}
	\caption{{The geometry of satellite network with a cluster area where satellites are randomly distributed on a sphere with radius $R_{\mathrm S}$, and the typical user is located on a sphere with radius $R_{\mathrm E}$.}}\label{fig_system}
 %between the edge of the cluster area and the line expending from the Earth's center to the zenith of the typical user.%
\end{figure}

\textbf{Remark 1.} {The cluster area $\mathcal{A}_{\mathrm{clu}}$ is designed to strategically group satellites geographically close to users and maintain proximity among satellites. Proximity to the user enhances channel gain, while proximity among satellites provides efficient data routing. This configuration allows {satellites without direct connection to a gateway to efficiently receive data via} multiple inter-satellite link (ISL) relays, particularly in scenarios with limited gateways. Unlike the $\mathrm{K}$-nearest satellites clustering approach {in \cite{Kim:2024}}, which requires frequent recalculations of proximity and may include distant satellites depending on the value of $\mathrm{K}$, {the cluster area-based approach} streamlines operations by requiring satellites to maintain only essential positional information and strictly limits inclusion to geographically proximate satellites.}

{\textbf{Remark 2.} When the typical user's nearby user forms a cluster area, the satellite cluster areas of the typical user and the nearby user may overlap, requiring satellites to serve multiple users {simultaneously}. This can reduce the available time-frequency resources per satellite and {weaken} the cooperation gain. However, since satellite networks operate over large bandwidths, such as the Ka- and Ku-bands \cite{3GPP:821}, satellites in overlapping regions {may be able to} assign orthogonal time-frequency resources based on users' data rate requirements in practical applications (e.g., text messaging, internet browsing). As a result, cooperation gain can be maintained, and the typicality assumption under the Slivnyak's theorem remains valid without loss of generality.}

\subsection{Channel Model}
\subsubsection{Propagation Model}
Our model encompasses large- and small-scale fading to characterize channel quality variation. Large-scale fading is captured using a path loss model, denoted as $\lVert \mathbf{x}_{l}-\mathbf{u}_1\rVert^{-\alpha}$ for $l \in \{1,\ldots,{\mathsf N}_{\mathrm S}\}$, which depends on the distance and the path loss exponent $\alpha$ \cite{Okati:2020}, \cite{Park:2023}. {For small-scale fading, we model the LOS properties, attributed to the high altitude of the satellites, using the Nakagami-$m$ distribution \cite{Chae:2023, Noh:2023, Okati:2022, Park:2023, Okati:2023}. This distribution captures LOS and non-line-of-sight (NLOS) effects by adjusting the Nakagami parameter~$m$, where $m=1$ corresponds to the Rayleigh distribution, and $m\rightarrow\infty$ approximates no fading. The PDF of the power of small-scale channel fading $\mathsf{H}_l$ from the satellite $l$ to the typical user is expressed as 
\begin{align} \label{nak}
f_{\mathsf{H}_l}(h)=\frac{m^m e^{-m h} h^{m-1}}{\Gamma(m)}, \quad h\geq 0, 
\end{align}
where $\mathbb{E}[\mathsf{H}_l]=1$ and $\mathbb{E}[\mathsf{H}_l^2]=\left(1+\frac{1}{m}\right)$.

{While the Nakagami-$m$ distribution effectively captures LOS effects, the shadowed-Rician distribution is also used to characterize small-scale fading in satellite networks \cite{Jung:2022}. The PDF of the squared magnitude of the shadowed-Rician fading is given by
\begin{align}
f_{\mathsf{H}_l}(h)=\frac{1}{2b_{0}}\left(\frac{2mb_{0}}{2mb_{0}+\Omega}\right)^{m}\exp\left(-\frac{h}{2b_{0}}\right) \notag \\
&\hspace{-11 pc} \cdot {}_1F_1 \left(m,\ 1,\ \frac{\Omega}{2b_{0}(2mb_{0}+\Omega)} h \right), 
\end{align}
where $b_0$ represents half of the average power of the NLOS component, $\Omega$ denotes the average power of the LOS component, and $m$ is the Nakagami parameter.}

{Although both the Nakagami-$m$ and the shadowed-Rician fading effectively characterize LOS effects in satellite networks, we adopt the Nakagami-$m$ distribution in this study for analytical tractability in performance analysis. Fig.~\ref{comp_naksha} of Section {\romanfour} exhibits moderate alignment in performance under these two small-scale fading models, indicating that our framework can be extended for performance analysis when the shadowed-Rician fading is considered.}

\subsubsection{Antenna Gain}
We assume that the satellites' beams have directional radiation patterns, while the antenna of the typical user employs an isotropic radiation pattern. 
%Satellites maintain the boresight of their beams in the direction of the sub-satellite point. 

%This model can explain two overlapping scenarios: one happens due to the satellites within the cluster area have overlapped coverage to the ground user and the other happens due to the overlapping of the cluster area because of the impact of another ground user's satellite cluster area to the typical user's satellite cluster area.

{Given that the satellite employs a directional beam radiation pattern focused in the sub-satellite direction to maximize signal power, the beam pattern can be effectively approximated using a sectored model with constant directivity gains for both main and side lobes, as explained in \cite{Bai:2014}. This approximation preserves critical characteristics of complex beamforming patterns, such as directivity gain and half-power beamwidth, while significantly simplifying the analysis. The sectored antenna gain for the satellite $l$, based on the two-lobe approximation with Dolph-Chebyshev beamforming weights \cite{Kor:2009} and Friis's equation \cite{Friis}, is expressed as \cite{Park:2023},\cite{Bai:2015, Singh:2015, Renzo:2015}
\begin{align}
G_{l} = \begin{cases} \displaystyle G_{\mathrm i}^{\mathrm{ t}}G^{\mathrm{ r}} \frac {c^{2}}{(4\pi f_{\mathrm c})^{2}}, & \mathbf{x}_{l} \in \mathcal{A}_\mathrm{clu},\\ \displaystyle G_{\mathrm o}^{\mathrm{ t}}G^{\mathrm{ r}}\frac {c^{2}}{(4\pi f_{\mathrm c})^{2}}, & \mathbf{x}_{l} \in \mathcal{A} \backslash \mathcal{A}_\mathrm{clu}, \end{cases}
\end{align}
where $f_{\mathrm c}$ is the carrier frequency, and $c$ is the speed of light.} The transmit antenna gain is determined by the satellite location, with $G^{\mathrm t}_{\mathrm i}$ for satellites in the cluster area and $G^{\mathrm t}_{\mathrm o}$ for those outside.
%\black{Due to the intersections across multiple orbital planes and increasing deployment of LEO satellites, proximity among satellites within the cluster area {occurs frequently} {\cite{Su:2019,Cao:2024}}. This proximity underpins the antenna gain model's assumption that near satellites can create coverage overlap, thereby covering a terrestrial area where the typical user is located with multiple beams from adjacent satellites.} \black{Moreover, defining a cluster area $\mathcal{A}_{\mathrm{clu}}$ {converts} potential sources of interference into components of a clustered satellite network, ensuring that satellites with main beamforming gain serve cooperatively rather {than} as sources of interference. Satellites outside the cluster area, which direct their main beams towards the subsatellite point and maintain significant distance from the user and their main beam, are then categorized as sources of interference, leading to significantly reduced antenna gain.}
The typical user's receive antenna gain is denoted as $G^{\mathrm r}$. 
%In this sectored model, antenna gains are assumed to be constant in the main lobe and the side lobe respectively for mathematical tractability.
%\textcolor{red}{We consider the scenario that the antenna gain $G_l$ remains constant for both $l \in \mathcal{A}_\mathrm{clu}$ and $l \in \mathcal{A} / \mathcal{A}_\mathrm{clu}$. This results from the fact that the cluster area needs not be extensive, as there are numerous satellites in orbit capable of serving the user cooperatively, and the beam gains of satellites located outside the cluster area are smaller than the beams gains of the serving satellites.}%

The sectored antenna gain model accounts for coverage overlap between satellite beams targeting users due to the satellite proximity. Given that satellite proximity within the cluster area occurs frequently \cite{Su:2019, Cao:2024}, driven by intersections across multiple orbital planes and the increasing deployment of LEO satellites, this proximity contributes to coverage overlap of each satellite's beam to the user, thereby enabling a terrestrial area where users are located to be served by multiple beams from adjacent satellites. \black{According to the sectored antenna gain model, adjacent satellites within the cluster area then \black{will cooperatively serve the typical user} with main-lobe gains, while satellites outside the cluster area \black{will cause interference to the typical user with} side-lobe gains, resulting in reduced antenna gain due to the large off-boresight angles to the user \cite{Tang:2021}.}  In cases where the cluster area of another user overlaps with the typical user's cluster, satellites in the overlapping region will serve both users simultaneously with main-lobe gains due to the proximity between the users.

\subsection{Transmission Method and Performance Metric}\label{metric}

%\section{Mathematical Preliminaries}\label{sec3}
\black{{Orthogonal frequency division multiplexing (OFDM)}, adopted in 3rd Generation Partnership Project (3GPP) standards for 5G New Radio (NR), has inspired its application in satellite to enhance compatibility with terrestrial networks \cite{Huang:2021, Nein:2024, Nein:jsac, Hee:2011}. We explore OFDM-based satellite networks where satellites in the cluster area cooperatively serve the typical user through non-coherent joint transmission. {This method enables satellites to serve the typical user without stringent synchronization, enhancing received power via non-coherent signal summation. We refer to \cite{Tan:2014} for details about the non-coherent joint transmission.}}

We assume an interference-limited network and use the signal-to-interference ratio (SIR) as the performance metric to derive coverage probability. This choice encapsulates the scenario where interference becomes predominant with the increasing number of satellites, rendering noise negligible. We also approximate the accumulated received power from the satellites within and outside the cluster area as Gamma random variables, consolidating these compound terms into their respective single random variables.

%We consider an orthogonal frequency division multiplexing (OFDM)-based satellite network with satellites in the cluster area transmitting signals using non-coherent joint transmission to serve the typical user.
%\subsection{Performance Metric} \label{metric}
The SIR at the typical user, located at $\mathbf{u}_1=(0,0,R_{\mathrm  E})$, is expressed as \black{[15, Appendix A]}
\begin{align} \label{SIR}
\mathrm{SIR}&=\frac{\sum_{\mathbf{x}_{l} \in \Pi \cap \mathcal {A}_{\mathrm{clu}}} G^{\mathrm t}_{\mathrm i} P {\mathsf H}_l||\mathbf{x}_{l}-\mathbf{u}_1||^{-\alpha}
}{\sum_{\mathbf{x}_{l} \in \Pi \cap \bar{\mathcal{A}}_{\mathrm{clu}}} G^{\mathrm t}_{\mathrm o} P {\mathsf H}_l ||\mathbf{x}_{l}-\mathbf{u}_1||^{-\alpha}},
%\frac{\sum_{ \mathbf{x}_l \in \Pi \cap \mathcal {A}_{\mathrm{clu}}} H_l||\mathbf{x}_l-\mathbf{u}_1||^{-\alpha}}{\sum_{\mathbf{x}_l \in \Pi \cap \bar{\mathcal{A}}_{\mathrm{clu}}} \bar{G}_o H_l ||\mathbf{x}_l-\mathbf{u}_1||^{-\alpha}}%
\end{align}
where $\bar{\mathcal{A}}_{\mathrm{clu}}= \mathcal{A} \backslash \mathcal {A}_{\mathrm{clu}}$ denotes the outside of the cluster area, and $P$ represents the transmit power. In \eqref{SIR}, the numerator and denominator indicate the accumulated power from the satellites within and outside the cluster area. For analytical simplicity, we normalize these two terms by the transmit power~$P$. The accumulated cluster power is then characterized as
\begin{align}
\mathsf{D}=\sum_{ \mathbf{x}_{l} \in \Pi \cap \mathcal {A}_{\mathrm{clu}}} G^{\mathrm t}_{\mathrm i} {\mathsf H}_l||\mathbf{x}_{l}-\mathbf{u}_1||^{-\alpha},
\end{align}
while the accumulated interference power is given by
\begin{align}
\mathsf{I}=\sum_{\mathbf{x}_{l} \in \Pi \cap \bar{\mathcal{A}}_{\mathrm{clu}}} G^{\mathrm t}_{ \mathrm o} {\mathsf H}_l ||\mathbf{x}_{l}-\mathbf{u}_1||^{-\alpha}.
\end{align}

\subsection{Gamma Approximation} \label{gamma_approx}
Before deriving the coverage probability, we first approximate the statistics of the sum of random variables on the state space of a Poisson process. The suitability of the Gamma random variable for closely modeling the statistics of the accumulated interference power in terrestrial networks is substantiated by \cite{Haenggi:2009book}, \cite{Heath:2013}. Building upon these findings, we approximate both the accumulated interference power $\mathsf I$ and the accumulated cluster power $\mathsf D$ in the context of satellite networks.
\begin{proposition} \label{shape_and_scale}
The accumulated interference power $\mathsf I$ can be approximated as a Gamma random variable $\tilde{\mathsf I}$. The shape parameter $k_{{\tilde{\mathsf I}}}$ and scale parameter  $\theta_{{\tilde{\mathsf I}}}$ of $\tilde{\mathsf I}$ are obtained as 
\begin{align} \label{shape_I}
k_{\tilde{\mathsf I}}=\frac{4(\alpha-1)\pi \lambda_{\mathrm S} \frac{R_{\mathrm {S}}}{R_{\mathrm {E}}}}{(\alpha-2)^2\left(1+\frac{1}{m}\right)}\frac{(R_{\mathrm{clu}}^{-\alpha+2}-R_{\mathrm{max}}^{-\alpha+2})^2}{R_{\mathrm{clu}}^{-2\alpha+2}-R_{\mathrm{max}}^{-2\alpha+2}},
\end{align}

\begin{align} \label{scale_I}
\theta_{\tilde{\mathsf I}}=\frac{(\alpha-2)G^{\mathrm t}_{\mathrm o} }{2(\alpha-1)} \left(1+\frac{1}{m}\right)\frac{R_{\mathrm{clu}}^{-2\alpha+2}-R_{\mathrm{max}}^{-2\alpha+2}}{R_{\mathrm{clu}}^{-\alpha+2}-R_{\mathrm{max}}^{-\alpha+2}}.
\end{align}
\end{proposition}
\begin{proof}
See Appendix \ref{appendix_shape_and_scale}.
\end{proof}
\noindent We also consider the accumulated cluster power $\mathsf D$ as a Gamma random variable $\tilde{\mathsf D}$. The shape parameter $k_{\tilde{\mathsf D}}$ and scale parameter $\theta_{\tilde{\mathsf D}}$ \black{of $\tilde{\mathsf D}$} are given by
\begin{align}\label{shape_D}
k_{\tilde{\mathsf D}}=\frac{4(\alpha-1)\pi \lambda_{\mathrm S} \frac{R_{\mathrm {S}}}{R_{\mathrm {E}}}}{(\alpha-2)^2\left(1+\frac{1}{m}\right)}\frac{(R_{\mathrm{min}}^{-\alpha+2}-R_{\mathrm{clu}}^{-\alpha+2})^2}{R_{\mathrm{min}}^{-2\alpha+2}-R_{\mathrm{clu}}^{-2\alpha+2}},
\end{align}
\begin{align}\label{scale_D}
\theta_{\tilde{\mathsf D}}=\frac{(\alpha-2) G^{\mathrm t}_{\mathrm i} }{2(\alpha-1)}\left(1+\frac{1}{m}\right)\frac{R_{\mathrm{min}}^{-2\alpha+2}-R_{\mathrm{clu}}^{-2\alpha+2}}{R_{\mathrm{min}}^{-\alpha+2}-R_{\mathrm{clu}}^{-\alpha+2}}.
\end{align}
These parameters are obtained through the same process as in the aforementioned proposition. {As shown in \eqref{shape_I}-\eqref{scale_D}, the distributions of $\tilde{\mathsf D}$ and $\tilde{\mathsf I}$ are influenced by $R_{\mathrm{clu}}$, which determines the cluster area.}

\black{The impact of $R_{\mathrm{clu}}$ on {the accumulated cluster power and interference power} is elucidated through the derivative of $k_{\tilde{\mathsf D}}$, $\theta_{\tilde{\mathsf D}}$, $k_{\tilde{\mathsf I}}$, and $\theta_{\tilde{\mathsf I}}$, with respect to $R_{\mathrm{clu}}$. Initially, the derivative for $k_{\tilde{\mathsf D}}$ and its sign are derived as follows:
\begin{align} \label{diff_shape}
\frac{dk_{\tilde{\mathsf{D}}}}{dR_{\mathrm{clu}}}&=\frac{4(\alpha-1)\pi \lambda_{\mathrm S} \frac{R_{\mathrm {S}}}{R_{\mathrm {E}}}}{(\alpha-2)^2\left(1+\frac{1}{m}\right)} \frac{R_{\mathrm{min}}^{-\alpha+2}-R_{\mathrm{clu}}^{-\alpha+2}}{R_{\mathrm{min}}^{-2\alpha+2}-R_{\mathrm{clu}}^{-2\alpha+2}} \notag \\
&\hspace{-2.5 pc} \cdot \left(2(\alpha-2)R_{\mathrm{clu}}^{-\alpha+1}-\frac{R_{\mathrm{min}}^{-\alpha+2}-R_{\mathrm{clu}}^{-\alpha+2}}{R_{\mathrm{min}}^{-2\alpha+2}-R_{\mathrm{clu}}^{-2\alpha+2}}(2\alpha-2)R_{\mathrm{clu}}^{-2\alpha+1}\right) \notag \\
&\hspace{0 pc}=\frac{4(\alpha-1)\pi \lambda_{\mathrm S} \frac{R_{\mathrm {S}}}{R_{\mathrm {E}}}}{(\alpha-2)^2\left(1+\frac{1}{m}\right)}\frac{2R_{\mathrm{clu}}^{-2\alpha+1}\left(R_{\mathrm{min}}^{-\alpha+2}-R_{\mathrm{clu}}^{-\alpha+2}\right)}{\left(R_{\mathrm{min}}^{-2\alpha+2}-R_{\mathrm{clu}}^{-2\alpha+2}\right)^2} \notag \\
&\hspace{-2.5 pc} \cdot \underbrace{\left(R_{\mathrm{clu}}^{-\alpha+2}+(\alpha-2)R_{\mathrm{clu}}^{\alpha}R_{\mathrm{min}}^{-2\alpha+2}- (\alpha-1) R_{\mathrm{min}}^{-\alpha+2}\right)}_{{\text{X}}}.
\end{align}
The term X in \eqref{diff_shape} increases with $R_{\mathrm{clu}}$ since its derivative is positive, i.e., $\frac{d\text{X}}{dR_{\mathrm{clu}}}=(\alpha-2)R_{\mathrm{clu}}^{-\alpha+1}\left(\alpha(\frac{R_{\mathrm{clu}}}{R_{\mathrm{min}}})^{2\alpha-2}-1\right) >0$ for $R_{\mathrm{clu}}>R_{\mathrm{min}}$. Additionally, the value of X at $R_{\mathrm{clu}}=R_{\mathrm{min}}$, which marks its minimum value, is zero, verifying the positive sign of the term X. Thus, the sign of $\frac{dk_{\tilde{\mathsf{D}}}}{dR_{\mathrm{clu}}}$ is conclusively positive. This indicates that as the cluster area expands}, the shape parameter of $\tilde{\mathsf{D}}$ increases, {resulting in an increase in the average value of~$\tilde{\mathsf{D}}$}. \black{The derivative for $\theta_{\tilde{\mathsf D}}$ and its sign are derived as follows:
\begin{align} \label{diff_scale}
\frac{d \theta_{\tilde{\mathsf D}}}{dR_{\mathrm{clu}}}
&=\frac{(\alpha-2)G^{\mathrm t}_{\mathrm i} }{2(\alpha-1)} \left(1+\frac{1}{m}\right) \frac{1}{R_{\mathrm{min}}^{-\alpha+2}-R_{\mathrm{clu}}^{-\alpha+2}} \notag \\
&\hspace{-2.5 pc} \cdot \left((2\alpha-2)R_{\mathrm{clu}}^{-2\alpha+1}-\frac{R_{\mathrm{min}}^{-2\alpha+2}-R_{\mathrm{clu}}^{-2\alpha+2}}{R_{\mathrm{min}}^{-\alpha+2}-R_{\mathrm{clu}}^{-\alpha+2}}(\alpha-2)R_{\mathrm{clu}}^{-\alpha+1}\right) \notag \\
&\hspace{0 pc}=\frac{(\alpha-2)G^{\mathrm t}_{\mathrm i} }{2(\alpha-1)} \left(1+\frac{1}{m}\right) \frac{R_{\mathrm{clu}}^{-\alpha+1}}{(R_{\mathrm{min}}^{-\alpha+2}-R_{\mathrm{clu}}^{-\alpha+2})^2} \notag \\
&\hspace{-2.5 pc} \cdot \underbrace{\left((2\alpha-2)R_{\mathrm{clu}}^{-\alpha} R_{\mathrm{min}}^{-\alpha+2} -\alpha R_{\mathrm{clu}}^{-2\alpha+2}-(\alpha-2) R_{\mathrm{min}}^{-2\alpha+2}\right)}_{{\text{Y}}}.
\end{align}
\noindent The term Y in \eqref{diff_scale} decreases with $R_{\mathrm{clu}}$ since its derivative is negative, i.e., $\frac{d\text{Y}}{dR_{\mathrm{clu}}}=\alpha (2\alpha-2)R_{\mathrm{clu}}^{-2\alpha+1}\left(1-(\frac{R_{\mathrm{clu}}}{R_{\mathrm{min}}})^{\alpha-2}\right)<0$ for $R_{\mathrm{clu}}>R_{\mathrm{min}}$. Furthermore, the value of Y at $R_{\mathrm{clu}}=R_{\mathrm{min}}$, which represents its maximum value, is zero, confirming the negative sign of the term Y. Consequently, the sign of $\frac{d\theta_{\tilde{\mathsf{D}}}}{dR_{\mathrm{clu}}}$ is conclusively negative. This indicates that, as the cluster area expands, the scale parameter of {$\tilde{\mathsf{D}}$} decreases, leading to a distribution more concentrated around its mean.}

\black{Derivatives in \eqref{diff_shape} and \eqref{diff_scale} support the intuition that expanding the cluster area increases the accumulated cluster power by incorporating a larger number of satellites. This increase is not solely due to the additional satellites but also includes satellites that previously caused interference. Therefore, the expansion of the cluster area achieves both an increase in accumulated cluster power and a reduction in accumulated interference power. Similarly, analyzing the impact of expanding cluster area on the accumulated interference power using shape and scale parameters of $\tilde{\mathsf{I}}$ reveals that $k_{\tilde{\mathsf I}}$ and~$\theta_{\tilde{\mathsf I}}$ decrease as the cluster area expands, resulting in reduced accumulated interference power, which aligns with our {aforementioned} discussion and intuition.}

\black{{\textbf{Remark 3.}} An increase in $R_{\mathrm{clu}}$ enlarges the cluster area, improving the numerator and reducing the denominator of SIR due to enhanced cooperation among satellites within the increased cluster area. {However, this cooperation gain entails additional costs, such as {increased system complexity and a reduced probability of a user being accepted by a satellite due to limited resources}. To account for these impacts, postulating cost models and integrating them into performance analysis for {optimizing overall network performance} is necessary, representing a promising direction for future research.}}

\black{{\textbf{Remark 4.}} The Laplace transforms of the accumulated cluster power and interference power, {$\mathcal {L}_{\{\mathsf D\}}(s)=\mathbb{E}[e^{-s \mathsf{D}}]$} and {$\mathcal {L}_{\{\mathsf I\}}(s)=\mathbb{E}[e^{-s \mathsf{I}}]$}, respectively decrease and increase as the cluster area expands. These changes raise the lower and upper bounds of the coverage probability, thereby increasing the coverage probability positioned between them. This is evidenced by the reduction of the $n=0$ term in \eqref{cov_bound} caused by the decreasing Laplace transform of $\mathsf{D}$ and the increase in the $n=0$ term in \eqref{cov_bound_2} resulting from the increasing Laplace transform of $\mathsf{I}$.}
%\blue{as demonstrated in Theorem~\ref{cov} and~\ref{cov_clu_theorem} in the next section.}

The coverage probability is derived using the statistical properties of \black{the Gamma random variables $\tilde{\mathsf I}$ or $\tilde{\mathsf D}$}. When deriving the coverage probability based on $\tilde{\mathsf I}$, \black{which approximates the accumulated interference power as the Gamma random variable}, we can directly incorporate \black{the accumulated cluster power $\mathsf D$} without employing Gamma approximation. On the contrary, when the coverage probability is based on $\tilde{\mathsf D}$, \black{which approximates the accumulated cluster power as the Gamma random variable}, we can consider \black{the accumulated interference power $\mathsf I$} directly without its approximation. For the following section, we elaborate on the coverage probability using $\tilde{\mathsf I}$.

\section{Coverage Probability}\label{sec3}
In this section, we propose the coverage probability of the satellite cluster network, with a focus on downlink satellite cooperation in the cluster area. We utilize the SIR that incorporates joint transmissions in the cluster area, as defined in Section \ref{metric}, and apply the Gamma approximation outlined in Section \ref{gamma_approx} to derive the lower and upper bounds for the coverage probability. \black{Given the distinct spatial regions between within and outside the cluster area, the mutual independence of the processes governing the distribution of satellites in these regions pertains.}
%We assume the mutual independence of the processes governing the distribution of satellites within and outside the cluster area, given their distinct spatial regions.}

The coverage probability in the interference-limited regime is the probability that the SIR at the typical user is greater than or equal to a threshold $\gamma$. This probability is derived as 
\begin{align}
P^{\mathrm{cov}}(\gamma; \lambda_{\mathrm S}, R_{\mathrm {S}}, \phi_{\mathrm {clu}},\alpha,m) &=P(\mathrm{SIR}\geq \gamma ) \notag\\
&=P\left(\mathsf I\leq \frac{\mathsf D}{\gamma }  \right) \label{cov_inter}\\
&=P\left(\mathsf D \geq {\gamma  \mathsf I} \right)\label{cov_cluster},
\end{align}
%&=P\left(\frac{\mathsf D}{ \mathsf I} \geq \gamma \right) \notag\\
where $\gamma$ is the SIR threshold. \black{{As per the definition, the coverage probabilities in Fig.~\ref{fig_comp_single} empirically demonstrate the performance enhancement of the satellite cluster network. Multiple satellites within the cluster, cooperatively serving the typical user, mitigate inter-satellite interference more efficiently than the single nearest satellite, with this enhancement increasing as the number of satellites grows.} Albeit the SIR coverage probability effectively captures the overall benefits of the cooperation gain, it does not account for the costs associated with satellite cooperation. {However, with well-defined cost metrics, the coverage probability framework facilitates an analysis of the trade-off introduced by satellite cooperation, as it serves as a basis for deriving practical spectral efficiency. Incorporating a well-defined cost parameter $C$, the practical spectral efficiency is expressed as \cite{Lee:2015}:
\begin{align}
\mathrm{SE}(C;\lambda_{\mathrm S}, R_{\mathrm {S}}, \phi_{\mathrm {clu}},\alpha,m) \notag \\
&\hspace{-8 pc}= (1 - C) \int_{0}^{\infty} \log_2(1 + \gamma) \, dP^{\mathrm{cov},c}(\gamma; \lambda_{\mathrm S}, R_{\mathrm {S}}, \phi_{\mathrm {clu}},\alpha,m) \notag \\
&\hspace{-8 pc}= (1 - C) \int_{0}^{\infty} \frac{P^{\mathrm{cov}}(\gamma; \lambda_{\mathrm S}, R_{\mathrm {S}}, \phi_{\mathrm {clu}},\alpha,m)}{(1 + \gamma) \ln 2} \, d\gamma,
\end{align}
where $P^{\mathrm{cov},c} =P(\mathrm{SIR} < \gamma )$ denotes the complement of the coverage probability $P^{\mathrm{cov}}$.}}

\begin{figure} [t!]
	\centering
	\includegraphics[width=0.9\columnwidth]{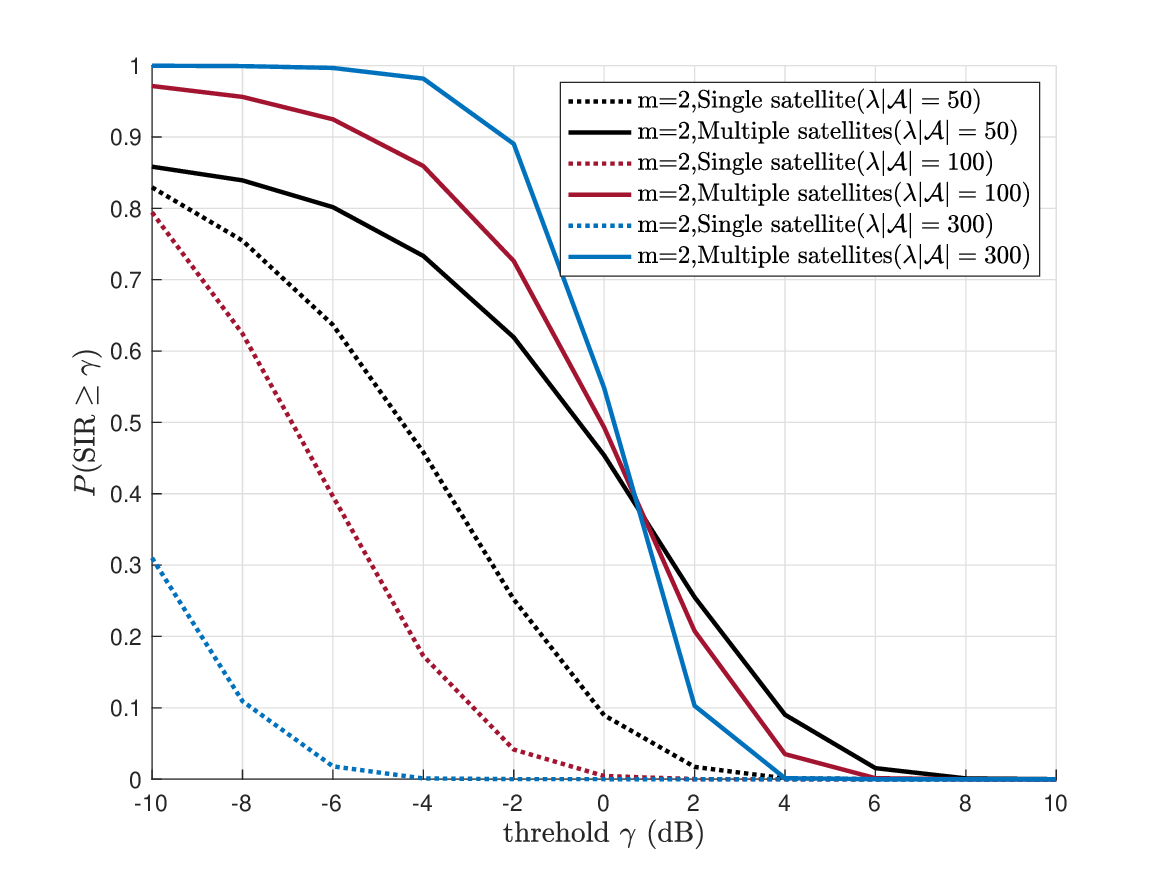}
	\caption{\black{Comparison of coverage probabilities when only the nearest satellite serves {the typical user} versus when multiple satellites in a cluster area serve {the typical user}, evaluated across $m=2$ and $\lambda_{\mathrm S}|\mathcal{A}| = 50, 100, 300$}.}\label{fig_comp_single}
\end{figure}

To derive the coverage probability based on \eqref{cov_inter}, we use the parameters given in \eqref{shape_I} and \eqref{scale_I} to consider the CDF of~$\tilde{\mathsf I}$. Similarly, when obtaining the coverage probability based on \eqref{cov_cluster}, we utilize the parameters provided in \eqref{shape_D} and \eqref{scale_D} to incorporate the complementary CDF of $\tilde{\mathsf D}$. For the remainder of this section, we opt for \eqref{cov_inter} to attain the coverage probability with the approximated Gamma random variable $\tilde{\mathsf I}$.

%which is suitable when a tight bound is required with a relatively small number of satellites. However, for a large number of satellites, it is preferred to derive coverage probability defined as (15) based on the Gamma approximation for the accumulated cluster power D.
In order to obtain the coverage probability as per \eqref{cov_inter}, it is necessary to compute the Laplace transform of {$\mathsf D$}. This process commences with obtaining the conditional Laplace transform of $\mathsf D$, which is then extended to the unconditional case through marginalization over the Poisson distribution. Before obtaining the conditional Laplace transform, we introduce the PDFs of the distance between the typical user and a satellite located in the observable spherical dome $\mathcal{A}$ in the following lemma.
\begin{lemma} \label{pdf_dis}
When conditioned on whether the satellite $l$ is located within or outside the cluster area, i.e., $\mathbf{x}_l \in \mathcal{A}_{\mathrm{clu}}$ or $\mathbf{x}_l \in \bar{\mathcal{A}}_{\mathrm{clu}}$, the PDFs of the distance $\mathsf R=\lVert \mathbf{x}_{l}-\mathbf{u}_1\rVert$ are given~by
\begin{align} \label{pdf_clu}
f_{\mathsf R | \mathbf{x}_l \in \mathcal{A}_{\mathrm{clu}}}(r) = \frac{r}{R_\mathrm  E R_\mathrm  S (1-\cos\phi_{\mathrm{clu}})},
\end{align}
for $R_{\mathrm{min}} \leq r \leq R_{\mathrm{clu}}$, and 
\begin{align}
f_{\mathsf R | \mathbf{x}_l \in \bar{\mathcal{A}}_{\mathrm{clu}} }(r) \notag \\
&\hspace{-4 pc}= \frac{r}{ R_\mathrm  E \left(R_\mathrm  S \cos\phi_{\mathrm{clu}}-R_\mathrm  E-R_{\mathrm{max}}\cos(\pi/2-\theta_{\mathrm{min}}) \right) },
\end{align}
for $R_{\mathrm{clu}} < r \leq R_{\mathrm{max}}$, respectively.
\end{lemma}
\begin{proof}
To derive the conditional PDF of $\mathsf R$, we let $\mathcal{A}(\mathsf R)$, $\mathsf R\in[R_{\mathrm{min}}, R_{\mathrm{max}}]$, denote a spherical dome whose maximum distance to the typical user is $R$. 
Using the fact that the surface area of $\mathcal{A}(\mathsf R)$, 
\begin{align}
|\mathcal{A}(\mathsf R)|=2\pi R_\mathrm  S \left(R_\mathrm  S-R_\mathrm  E -\frac{R_\mathrm  S^2-R_\mathrm  E^2 -\mathsf R^2}{2 R_\mathrm  E}\right),
\end{align}
is an increasing function of $\mathsf R$, the probability that $\mathsf R$ is less than or equal to $r$ is the same as the probability that 
the satellite $l$ is located in $\mathcal{A}(r)$. With this property,
% $|\mathcal{A}(\mathsf R)|$ is less than or equal to $|\mathcal{A}(r)|$. With this property,
the CDF of $\mathsf R$ for $R_{\mathrm{min}} \le r \le R_{\mathrm{clu}}$ is derived as 
\begin{align} \label{cdf_1}
% F_{\mathsf R| \mathbf{x}_l \in \mathcal{A}_{\mathrm{clu}} }(r) 
% &=P(\mathsf R \leq r| \mathbf{x}_l \in \mathcal{A}_{\mathrm{clu}}) \notag \\
% &=\frac{P(\mathsf R \leq r, \mathsf R  \le R_{\mathrm{clu}})}{P(\mathsf R  \le R_{\mathrm{clu}})} \notag \\
% &=\frac{P(\mathsf R \leq r)}{P(\mathsf R  \le R_{\mathrm{clu}})}=\frac{P(|\mathcal{A}(\mathsf R)| \le |\mathcal{A}(r)|)}{P(|\mathcal{A}(\mathsf R)| \le |\mathcal{A}_{\mathrm{clu}}|)}.
F_{\mathsf R| \mathbf{x}_l \in \mathcal{A}_{\mathrm{clu}} }(r) 
&=P(\mathsf R \leq r| \mathbf{x}_l \in \mathcal{A}_{\mathrm{clu}}) \notag \\
&=\frac{P(\mathbf{x}_l \in \mathcal{A}(r), \mathbf{x}_l \in \mathcal{A}_{\mathrm{clu}})}{P(\mathbf{x}_l \in \mathcal{A}_{\mathrm{clu}})} \notag \\
&=\frac{P(\mathbf{x}_l \in \mathcal{A}(r))}{P(\mathbf{x}_l \in \mathcal{A}_{\mathrm{clu}})}.
% =\frac{P(\mathbf{x}_l \in \mathcal{A}(r))}{P(\mathbf{x}_l \in \mathcal{A}_{\mathrm{clu}})}.
% &= \begin{cases} 
    %  0, & \mbox{if  } r < R_{\mathrm{min}} ,\\
    % \frac{P(\mathsf R \leq r)}{P(\mathsf R  \le R_{\mathrm{clu}})}, & \mbox{if  } R_{\mathrm{min}} \le r < R_{\mathrm{clu}},\\
    % 1, & \mbox{otherwise} 
    % \end{cases}
% &=P(|\mathcal{A}_{\mathsf R}| \le |\mathcal{A}_{r}| \,| \mathbf{x}_l \in \mathcal{A}_{\mathrm{clu}}) \notag\\
% &=\frac{P(|\mathcal{A}_{\mathsf R}| \le |\mathcal{A}_{r}|, \mathbf{x}_l \in \mathcal{A}_{\mathrm{clu}}) }{P(\mathbf{x}_l \in \mathcal{A}_{\mathrm{clu}})}
% &=P\left(\mathsf{A} < \frac{\pi {R_\mathrm  S}(r^2-{R_\mathrm  S}^2-{R_\mathrm  E}^2 + 2{R_\mathrm  S}{R_\mathrm  E})}{{R_\mathrm  E}}\right).
\end{align}
Given that a satellite distributed according to a PPP has the same probability of existence at any position,
% As the probability that a satellite exists within a certain area is the same for all satellites, 
the probability that the satellite $l$ is located in $\mathcal{A}(r)$, i.e., $P(\mathbf{x}_l \in \mathcal{A}(r))$,
% $P(|\mathcal{A}(\mathsf R)| \le |\mathcal{A}(r)|)$ 
is obtained as the ratio between the surface areas of $\mathcal{A}(r)$ and the whole sphere, i.e., $\frac{|\mathcal{A}(r)|}{4\pi R_\mathrm  S^2}$. With this, the CDF is given by
\begin{align} \label{cdf_final_1}
F_{\mathsf R| \mathbf{x}_l \in \mathcal{A}_{\mathrm{clu}}}(r) 
\!=\!\frac{|\mathcal{A}(r)|}{|\mathcal{A}_{\mathrm{clu}}|}
\!=\!\frac{r^2-R_\mathrm  S^2-R_\mathrm  E^2 + 2{R_\mathrm  S}{R_\mathrm  E}}{2{R_\mathrm  E}\left(R_\mathrm  S-R_\mathrm  E -\frac{R_\mathrm  S^2-R_\mathrm  E^2 -R_{\mathrm{clu}}^2}{2 R_\mathrm  E}\right) }.
\end{align}
Similarly, the conditional CDF of $\mathsf R$ for $R_{\mathrm{clu}} < r \leq R_{\mathrm{max}}$ is given by
\begin{align} \label{cdf_final_2}
F_{\mathsf R| \mathbf{x}_l \in \bar{\mathcal A}_{\mathrm {clu}}}(r) \notag\\
&\hspace{-4.5 pc}=\frac{r^2-R_\mathrm  {clu}^2}{2{R_\mathrm  E}\left(\frac{R_{\mathrm  S}^2-R_{\mathrm  E}^2-R_{\mathrm{clu}}^2}{2 R_\mathrm  E}-R_{\mathrm{max}}\cos\left(\pi/2 -\theta_{\mathrm{min}}\right)\right) }.
\end{align}
The corresponding PDFs can be derived by differentiating \eqref{cdf_final_1} and \eqref{cdf_final_2} with respect to $r$, which completes the proof.
%%%%% DH end %%%%%%
\end{proof}
From the result, we derive the conditional Laplace transform of $\mathsf D$. Let $\mathsf{N}_{\mathrm{S,clu}}$ denote the number of satellites in the cluster area $\mathcal{A}_{\mathrm{clu}}$. Fixing the number of cooperating satellites at $\mathsf{N}_{\mathrm{S,clu}}=L$, the locations of the $L$ satellites can be modeled using a BPP \cite{Sto:2013}. Then the conditional Laplace transform is given in the following lemma.
\begin{lemma} \label{lap_clu}
The conditional Laplace transform of the accumulated cluster power $\mathsf D$ is
\begin{align} \label{lap_clu_int}
	\mathcal {L}_{\{\mathsf D|\mathsf{N}_{\mathrm{S,clu}}=L\}}(s)=\Bigg(&\frac{1}{R_\mathrm  E R_\mathrm  S (1-\cos\phi_{\mathrm{clu}})\alpha} \notag \\
 &\hspace{-0.5 pc} \cdot \int_{R_{\mathrm{clu}}^{-\alpha}}^{R_{\mathrm{min}}^{-\alpha}} t^{-\frac{2}{\alpha}-1}\left(1+\frac{s G^{\mathrm t}_{\mathrm i} t}{m}\right)^{-m}dt\Bigg)^L.
\end{align} 
\end{lemma}
\begin{proof}
See Appendix \ref{appendix_lap}.
\end{proof}

In solving the integral in the conditional Laplace transform, computation complexity may arise. However, assuming the Rayleigh fading for satellite channels, i.e., $m=1$, reduces the integral to a \black{{compact solution}}. The simplified conditional Laplace transform for $m=1$ is obtained in the following corollary.
\begin{corollary}
Using the integral result from \cite{integral}, a simpler expression for the conditional Laplace transform of the accumulated cluster power $\mathsf{D}$, specifically when $m=1$, is given~by 
\begin{align} \label{lap_m=1_closed}
\mathcal {L}^{m=1}_{\{\mathsf D|\mathsf{N}_{\mathrm{S,clu}}=L\}}(s) \notag\\
&\hspace{-5 pc}=\Bigg[\frac{1}{R_\mathrm{E}R_\mathrm{S}(1-\cos\phi_{\mathrm{clu}})\alpha} \notag \\
&\hspace{-4.5 pc} \cdot \Bigg(\frac{R_{\mathrm{clu}}^{2+\alpha}}{sG_\mathrm{i}^{\mathrm t}(1+\frac{2}{\alpha})}
{}_2F_1\left(1, 1+\frac{2}{\alpha};2+\frac{2}{\alpha};-\frac{R_{\mathrm{clu}}^\alpha}{sG_\mathrm{i}^{\mathrm t}}\right) \notag\\
&\hspace{-4.5 pc}-\frac{R_{\mathrm{min}}^{2+\alpha}}{sG_\mathrm{i}^{\mathrm t}(1+\frac{2}{\alpha})} {}_2F_1\left(1, 1+\frac{2}{\alpha};2+\frac{2}{\alpha};-\frac{R_{\mathrm{min}}^\alpha}{sG_\mathrm{i}^{\mathrm t}}\right)\Bigg)\Bigg]^L.
\end{align}
\end{corollary}
\begin{proof}
By setting $m=1$ in \eqref{lap_clu_int}, the conditional Laplace transform of the accumulated cluster power is reduced to 
\begin{align} \label{lap_m=1}
\mathcal {L}^{m=1}_{\{\mathsf D|\mathsf{N}_{\mathrm{S,clu}}=L\}}(s) \notag\\
&\hspace{-4.5 pc}=\Bigg(\frac{1}{R_\mathrm{E}R_\mathrm{S}(1-\cos\phi_{\mathrm{clu}})\alpha}
\cdot \int_{R_{\mathrm{clu}}^{-\alpha}}^{R_{\mathrm{min}}^{-\alpha}}\frac{t^{-\frac{2}{\alpha}-1}}{1+sG_\mathrm{i}^{\mathrm t} t}dt \Bigg)^L.
\end{align}
The final expression in \eqref{lap_m=1_closed} can be obtained by {[48, Eq. 3.194.2.6]}, which completes the proof.
\end{proof}
%It can be seen that the closed-form equation allows for an intuitive examination of the system parameters influencing the coverage probability.

By marginalizing the conditional Laplace transform over $L$, the Laplace transform of the accumulated cluster power $\mathsf D$ is obtained as
\begin{align} \label{lap_cluster}
\mathcal{L}_{\{\mathsf {D}\}} (s) \notag\\
&\hspace{-1.5 pc}\stackrel{(a)}{=}\exp({-\lambda_{\mathrm S} {|\mathcal{A}_{\mathrm{clu}}|}}) \sum_{L=0}^{\infty} \frac{(\lambda_{\mathrm S} |\mathcal{A}_{\mathrm{clu}}|)^L}{L!}\mathcal{L}_{\{\mathsf D|\mathsf{N}_{\mathrm{S,clu}}=L\}}(s) \notag \\
&\hspace{-1.5 pc}\stackrel{(b)}{=}\exp \left(-\lambda_{\mathrm S} |\mathcal{A}_{\mathrm{clu}}| \left(1-\mathcal{L}^{\frac{1}{L}}_{\{\mathsf D|\mathsf{N}_{\mathrm{S,clu}}=L\}}(s)\right)\right),
\end{align}
where (a) is obtained using the Poisson distribution, and (b) is determined through the Taylor series expansion of the exponential function, $ e^x=\sum_{n=0}^{\infty} \frac{x^n}{n!}$. It is noteworthy that the conditional Laplace transform in \eqref{lap_cluster} becomes independent of $L$ through the ${L}$th root operation on \eqref{lap_m=1}.

Due to the non-integer property of \black{the shape parameters of $\tilde{\mathsf I}$ and $\tilde{\mathsf D}$, denoted $k_{\tilde{\mathsf I}}$ and $k_{\tilde{\mathsf D}}$ respectively}, obtaining the exact form of the coverage probability entails the utilization of the Gamma random variable's CDF, which involves infinite summations. These infinite operations impose computational complexity and make the mathematical analysis of the results challenging. To address this issue without compromising generality, we proceed to establish the lower and upper bounds of the coverage probability using the Erlang distribution while considering a finite number of terms $\tilde{k}$, \black{the rounded value of the shape parameter}, in both differentiation and summation \cite{Erlang}. Using the definition of the coverage probability \black{based on $\tilde{\mathsf I}$, which approximates the accumulated interference power as a Gamma random variable as outlined in \eqref{cov_inter}, and the Laplace transform of the accumulated cluster power $\mathsf {D}$, derived in \eqref{lap_cluster}}, we establish the coverage probability bounds in the following theorem. It is noteworthy that the lower and upper bounds, which will be derived shortly, coincide when the shape parameter, $k_{\tilde{\mathsf I}}$ or $k_{\tilde{\mathsf D}}$, is an integer.
\begin{theorem} \label{cov}
The lower and upper bounds of the coverage probability, derived from \eqref{cov_inter} with an integer parameter $\tilde{k}$, are obtained as
\begin{align}  \label{cov_bound}
P^{\mathrm{cov}}(\gamma; \lambda_{\mathrm S}, R_{\mathrm {S}}, \phi_{\mathrm{clu}},\alpha,m)&=\mathbb{E}\left[P\left({\mathsf I} \leq \frac{\mathsf D}{\gamma} \right)\right] \notag \\
&\hspace {-9 pc} {{\tilde{k}= \lceil k_{\tilde{\mathsf I}} \rceil\atop\geq}\atop{<\atop \tilde{k}=\lfloor k_{\tilde{\mathsf I}} \rfloor}} 1-\sum_{n=0}^{\tilde{k}-1}\frac{(\gamma\theta_{\tilde{\mathsf I}})^{-n}}{n!}(-1)^n \frac{d^n \mathcal{L}_{\{\mathsf D\}} (s)}{ds^n}\Bigg|_{s=\frac{1}{\gamma \theta_{\tilde{\mathsf I}}}}.
\end{align}
\end{theorem}
\begin{proof}
See Appendix \ref{appendix_cov}.
\end{proof}
The coverage probability based on \eqref{cov_cluster} then incorporates~$\tilde{\mathsf D}$, \black{which approximates the accumulated cluster power as a Gamma random variable}, and the Laplace transform of \black{the accumulated interference power $\mathsf I$}. More detailed results are presented in the following theorem.
\begin{theorem} \label{cov_clu_theorem}
The lower and upper bounds of the coverage probability using \eqref{cov_cluster} with an integer parameter $\tilde{k}$, are given~by
\begin{align}  \label{cov_bound_2}
P^{\mathrm{cov}}(\gamma; \lambda_{\mathrm S}, R_{\mathrm {S}}, \phi_{\mathrm{clu}},\alpha,m)&=\mathbb{E}\left[P\left(\mathsf D \geq {\gamma  \mathsf I}  \right)\right] \notag \\
&\hspace {-7 pc} {{\tilde{k}= \lfloor k_{\tilde{\mathsf D}} \rfloor\atop\geq}\atop{<\atop \tilde{k}=\lceil k_{\tilde{\mathsf D}} \rceil}}\sum_{n=0}^{\tilde{k}-1}\frac{\gamma^n}{n! \left(\theta_{\tilde{\mathsf D}}\right)^{n} }(-1)^n \frac{d^n \mathcal{L}_{\{\mathsf I\}} (s)}{ds^n}\Bigg|_{s=\frac{\gamma}{\theta_{\tilde{\mathsf D}}}}.
\end{align}
\end{theorem}
\begin{proof}
The proof can be easily derived using the same principles as in Theorem \ref{cov}.
\end{proof}
From the above results, we can see that the coverage probability depends on various system parameters, including threshold $\gamma$, satellite density $\lambda_{\mathrm S}$, satellite altitude $R_\mathrm {S}$, cluster area determined by $\phi_{\mathrm{clu}}$, path loss exponent $\alpha$, and the Nakagami parameter $m$. The primary distinction between adopting $\tilde{\mathsf I}$ or~$\tilde{\mathsf D}$ pertains to their respective shape parameters, which are determined by the difference in distances, as shown in \eqref{shape_I} and \eqref{shape_D}. Notably, the shape parameter for $\tilde{\mathsf I}$ is usually higher than that of $\tilde{\mathsf D}$ because the difference between $R_{\mathrm{clu}}$ and $R_{\mathrm{max}}$ is much larger than the difference between $R_{\mathrm{min}}$ and $R_{\mathrm{clu}}$ particularly when considering a small cluster polar angle~$\phi_{\mathrm{clu}}$. As shown, $\tilde{k}$ determines the extent of differentiation and summation required in the coverage probability calculation. The large shape parameter of $\tilde{\mathsf I}$ then narrows the gap between the lower and upper bounds of the coverage probability but increases the computational burden. Conversely, employing a small shape parameter with $\tilde{\mathsf D}$ can potentially result in a wider gap between the two bounds compared to $\tilde{\mathsf I}$, but it enables the analyses of network with a massive number of satellites while reducing the computational load. More details about this trade-off can be found in simulation results.
%\begin{remark}
%\black{The condition for the lower and upper bounds to coincide is that the shape parameter, $k_{\tilde{\mathsf I}}$ or $k_{\tilde{\mathsf D}}$, is an integer.}
%\end{remark}
%In terms of computational complexity, $\tilde{k}$ determines the extent of differentiation and summation required in the coverage probability calculation.
% However, due to the non-integer property of the shape parameters $k_{\tilde{\mathsf I}}$ and $k_{\tilde{\mathsf D}}$ in most cases, obtaining the exact form of the coverage probability entails the utilization of the Gamma random variable's CDF, which involves infinite summations. These infinite operations \black{impose} computational complexity and \black{render} the mathematical analysis of the results challenging. To release this issue without compromising generality, we proceed to establish the coverage probability bounds using the Erlang distribution while considering a finite number of terms $\tilde{k}$ in both differentiation and summation.
%The Laplace transform of $\mathsf I$ is obtained similarly to that of $\mathsf D$, with the only distinction being that satellites are considered outside the cluster area.

We further simplify the derivative of the Laplace transform of \black{the accumulated cluster power $\mathcal{L}_{\{\mathsf D\}} (s)$,} by leveraging the differentiation property of the exponential function. Utilizing Faà di Bruno's formula \cite{bell}, the $n$th derivative of $\mathcal{L}_{\{\mathsf D\}} (s)$ is expressed as
\begin{align} \label{result_bell}
\frac{d^n \mathcal{L}_{\{\mathsf D\}} (s)}{ds^n} =\sum_{q=1}^{n} \exp({\log\mathcal{L}_{\{\mathsf D\}}(s)}) \notag\\
&\hspace {-9.6 pc} \cdot B_{n,q}\left( \frac{d\log\mathcal{L}_{\{\mathsf D\}}(s)}{ds},\cdots, \frac{d^{n-q+1}\log\mathcal{L}_{\{\mathsf D\}}(s)}{ds^{n-q+1}} \right), 
\end{align}
where the Bell polynomial is defined as ${B_{n,q}(x_1,\ldots, x_{n-q+1})}$ ${=\sum \frac{n!}{a_1!\cdots a_{n-q+1}!} \left(\frac{x_1}{1!}\right)^{a_1} \cdots \left(\frac{x_{n-q+1}}{(n-q+1)!}\right)^{a_{n-q+1}}}$, and the sum is taken over all sequences $a_1, a_2, \ldots, a_{n-q+1}$ of non-negative integers satisfying ${a_1+a_2+\cdots+a_{n-q+1}=q}$.
%It is noteworthy that obtaining the derivation of the Laplace transform requires differentiation of its exponent.
Then, we can rewrite the lower and upper bounds of the coverage probability as follows
\begin{align} \label{cov_bell}
P^{\mathrm{cov}}(\gamma; \lambda_{\mathrm S}, R_{\mathrm {S}}, \phi_{\mathrm{clu}},\alpha,m) \notag \\
&\hspace {-9 pc} {{\tilde{k}= \lceil k_{\tilde{\mathsf I}} \rceil\atop\geq}\atop{<\atop \tilde{k}=\lfloor k_{\tilde{\mathsf I}} \rfloor}} 1- \Bigg[ \sum_{n=0}^{\tilde{k}-1}\frac{(\gamma\theta_{\tilde{\mathsf I}})^{-n}}{n!}(-1)^n \sum_{q=1}^{n} \exp({\log\mathcal{L}_{\{\mathsf D\}}(s)}) \notag \\
&\hspace {-9 pc}\cdot B_{n,q}\left( \frac{d\log\mathcal{L}_{\{\mathsf D\}}(s)}{ds},\cdots, \frac{d^{n-q+1}\log\mathcal{L}_{\{\mathsf D\}}(s)}{ds^{n-q+1}} \right) \Bigg|_{s=\frac{1}{\gamma \theta_{\tilde{\mathsf I}}}}\Bigg].
\end{align}

As can be seen in \eqref{cov_bell}, differentiating the exponent of the Laplace transform is now required to obtain the bounds of the coverage probability. For the special case under the Rayleigh fading, i.e., $m=1$, we simplify the $n$th derivative of the exponent of $\mathcal{L}^{m=1}_{\{\mathsf D\}} (s)$ in the following corollary.
\begin{corollary}
The $n$th derivative of the exponent of $\mathcal{L}^{m=1}_{\{\mathsf D\}} (s)$ is given by 
\begingroup
\allowdisplaybreaks
\begin{align} \label{diff_exp_m1}
\frac{d^n\log\mathcal{L}^{m=1}_{\{\mathsf D\}}(s)}{ds^n} \notag \\
&\hspace{-5.6 pc}=\lambda_{\mathrm S}|\mathcal{A}_{\mathrm{clu}}|\frac{n! (-G_{\mathrm i}^{\mathrm t})^n}{R_\mathrm E R_\mathrm S (1-\cos\phi_{\mathrm{clu}})\alpha} \notag \\
&\hspace{-5.2 pc}\cdot \Bigg[\frac{R_{\mathrm{clu}}^{2+\alpha}}{(sG_\mathrm{i}^{\mathrm t})^{n+1}(1+\frac{2}{\alpha})} {}_2F_1\left(n+1, 1+\frac{2}{\alpha};2+\frac{2}{\alpha};-\frac{R_{\mathrm{clu}}^\alpha}{sG_\mathrm{i}^{\mathrm t}}\right) \notag \\
&\hspace{-5.3 pc} -\frac{R_{\mathrm{min}}^{2+\alpha}}{(sG_\mathrm{i}^{\mathrm t})^{n+1}(1+\frac{2}{\alpha})}
{}_2F_1\left(n+1, 1+\frac{2}{\alpha};2+\frac{2}{\alpha};-\frac{R_{\mathrm{min}}^\alpha}{sG_\mathrm{i}^{\mathrm t}}\right)\Bigg].
\end{align}
\endgroup
\begin{proof}
The $n$th derivative of the exponent of $\mathcal{L}^{m=1}_{\{\mathsf D\}} (s)$ is obtained as
\begingroup
\allowdisplaybreaks
\begin{align} \label{diff_exp}
\frac{d^n\log\mathcal{L}^{m=1}_{\{\mathsf D\}}(s)}{ds^n} \notag \\
&\hspace{-5.5 pc}=\lambda_{\mathrm S}|\mathcal{A}_{\mathrm{clu}}| \frac{1}{R_\mathrm{E}R_\mathrm{S}(1-\cos\phi_{\mathrm{clu}})\alpha} \cdot \frac{d^n}{ds^n} \int_{R_{\mathrm{clu}}^{-\alpha}}^{R_{\mathrm{min}}^{-\alpha}}\frac{t^{-\frac{2}{\alpha}-1}}{1+sG_\mathrm{i}^{\mathrm t} t}dt \notag \\
&\hspace{-5.5 pc}\stackrel{(a)}{=}\lambda_{\mathrm S}|\mathcal{A}_{\mathrm{clu}}| \frac{1}{R_\mathrm{E}R_\mathrm{S}(1-\cos\phi_{\mathrm{clu}})\alpha} \cdot  \int_{R_{\mathrm{clu}}^{-\alpha}}^{R_{\mathrm{min}}^{-\alpha}} \frac{d^n}{ds^n} \frac{t^{-\frac{2}{\alpha}-1}}{1+sG_\mathrm{i}^{\mathrm t} t}dt \notag \\
&\hspace{-5.5 pc}=\lambda_{\mathrm S}|\mathcal{A}_{\mathrm{clu}}|\frac{n! (-G_{\mathrm i}^{\mathrm t})^n}{R_\mathrm E R_\mathrm S (1-\cos\phi_{\mathrm{clu}})\alpha} \cdot \int_{R_{\mathrm {clu}}^{-\alpha}}^{R_{\mathrm {min}}^{-\alpha}} \frac{t^{n-\frac{2}{\alpha}-1}}{(1+sG_{\mathrm i}^{\mathrm t}t)^{n+1}}dt,
\end{align}
\endgroup
where (a) follows from the Leibniz integration rule. The final expression in \eqref{diff_exp_m1} is obtained by {[48, Eq. 3.194.2.6]}, which completes the proof.   
\end{proof}
\end{corollary}
\noindent By invoking \eqref{diff_exp_m1} in \eqref{cov_bell}, {the lower and upper bounds} for the coverage probability can be derived for the case when $m=1$.

\section{Simulation Results {and Analyses}}\label{sec4}
In this section, we validate the Gamma approximation and the derived coverage probability bounds through Monte Carlo simulations. We also evaluate the impact of the system parameters on the coverage probability.

In our numerical experiments, we set the satellite altitude to $R_\mathrm{S}-R_\mathrm{E}=500$ km and restrict satellite visibility to $\theta_{\mathrm{min}}=25^\circ$, with a defined polar angle of the cluster area $\phi_{\mathrm{clu}} = 1.6^\circ$. {Based on the relationship between the radius of the cluster area $R_{\mathrm{clu}}$ (calculated using $\phi_{\mathrm{clu}}$) and the off-axis angle $\theta_{\mathrm{off-axis}}$ (which determines whether the antenna gain falls into the main-lobe or side-lobe, assuming satellites maintain the boresight of their beams toward the sub-satellite point), expressed as $R_{\mathrm{clu}}=R_{\mathrm{S}}\cos\theta_{\mathrm{off-axis}}-\sqrt{R_{\mathrm{S}}^2(-1+{\cos^2\theta_{\mathrm{off-axis}}})+R_{\mathrm{E}}^2}$, the radius of the cluster area is constrained by the off-axis angle. This constraint ensures that satellites within the cluster area {will} serve the typical user with main-lobe gains, while those outside the cluster area {will} provide side-lobe gains.} Utilizing a sectored antenna gain, we employ the normalized antenna gain $G_{\mathrm o}^{\mathrm{ t}}/G_{\mathrm i}^{\mathrm{t}} = 0.1$, indicating a 10 dB gain advantage for satellites in the cluster area compared to those outside. For the coverage probability calculations, we explore two scenarios: one with an average of $\lambda_{\mathrm S}|\mathcal{A}|= 50$ satellites in the observable spherical dome and the other with $\lambda_{\mathrm S}|\mathcal{A}|= 300$. Given these parameters, the average number of satellites on the satellite sphere, whose surface area is calculated as $4\pi R_{\mathrm  S}^2$, is at least 10k, and the average number of satellites within the cluster area exceeds two. Specific satellite numbers for these scenarios are detailed in Table \ref{table_1}.

For the analyses of two types of the derived coverage probability bounds, the bounds from the distribution of $\tilde{\mathsf{I}}$, the approximated Gamma random variable for the accumulated interference power $\mathsf I$, are assessed in scenario~1, while the bounds from the distribution of $\tilde{\mathsf{D}}$, the approximated Gamma random variable for the accumulated cluster power $\mathsf D$ is evaluated in scenario~2. The forthcoming results will offer a more detailed explanation and demonstrate how the utilization of these random variables achieves either elaborateness or computational efficiency. {Lastly, as demonstrated in \cite{Kim:2024}, the validity of our PPP-based modeling for coverage probability is substantiated by the reasonable alignment between the performance of satellite networks based on actual Starlink satellite distribution data and the PPP-based model for densely deployed constellations.}

%\red{Before analyzing the simulation results, we first justify using the Nakamgai-$m$ distribution for modeling small-scale fading. To validate this assumption, we compare the coverage probabilities under the Nakgami-$m$ and shadowed-Rician fading in Fig.~\ref{comp_naksha}. The simulation results reveal minimal differences between the two models, thereby supporting the adoption of Nakagami-$m$ fading in our analysis. Given that the shadowed-Rician distribution can be approximated by a Gamma distribution \cite{Abd:2003},\cite{Tal:2024}, and the squared magnitude of the Nakagami-$m$ random variable follows a Gamma distribution, our analytical framework can be extended to encompass the approximated shadowed-Rician fading, offering a tractable approach for satellite network performance analysis.}

{Before analyzing the simulation results, we compare the coverage probabilities under the Nakagami-$m$ fading, which is adopted in this study, and {the shadowed-Rician fading in Fig.~\ref{comp_naksha}}. The simulation results indicate {marginal} numerical differences between these two small-scale fading models, further substantiated by the mathematical fact that the shadowed-Rician fading can be effectively approximated by the Nakagami-$m$ fading. This approximation holds because the squared magnitude of the shadowed-Rician fading can be modeled as a Gamma random variable \cite{Abd:2003},\cite{Tal:2024}, while the squared magnitude of the Nakagami-$m$ fading inherently follows the Gamma distribution. Therefore, the proposed coverage probability analysis remains applicable to settings where the shadowed-Rician fading is incorporated into the analysis.}

%For sparse networks, non-homogeneous PPP or orbit geometry-dependent modeling would be more appropriate, which is out of the scope of this paper.

\begin{table} [t!]
\captionsetup{justification=centering, labelsep=newline, font={smaller,sc}}
\caption{Average Number of Satellites}
\centering
\begin{tabular}{|c|c|c|c|} 
\hline
Type & On $|\mathcal{A}|$ & On $|\mathcal{A}_{\mathrm{clu}}|$ & On satellite sphere\\
\hline
Scenario 1 & 50 & 2.0837 & 10,700\\
\hline
Scenario 2 & 300 & 12.5020 & 64,100\\
\hline
\end{tabular} \label{table_1}
\end{table}

\begin{figure} [t!]
	\centering
	\includegraphics[width=1\columnwidth]{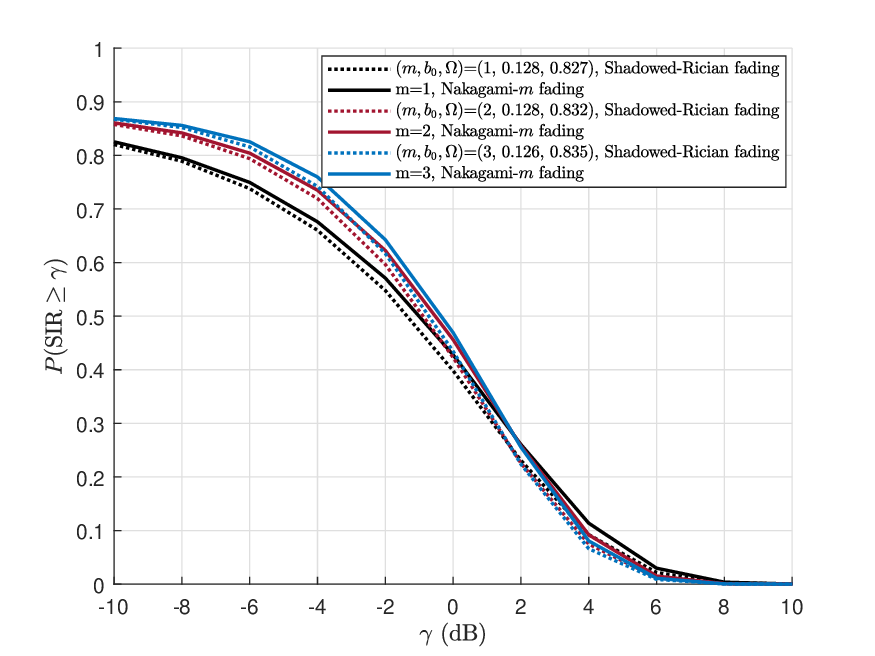}
	\caption{{Coverage probabilities under the shadowed-Rician fading with $(m,b_0,\Omega)=(1, 0.128, 0.827),$ $(2, 0.128, 0.832),$ $(3, 0.126, 0.835),$ corresponding to average shadowing \cite{Abd:2003}, and the Nakagami-$m$ fading with $m=1,2,3$, respectively.}}\label{comp_naksha}
\end{figure}

%In particular, two scenarios are examined for calculating coverage probabilities: one with 10,700 satellites and the other with 64,100 satellites, on average, distributed across the satellite sphere calculated as $4\pi R_{\mathrm  S}^2$. The first scenario, involving 10,700 satellites, employs the distribution of $\tilde{\mathsf{I}}$, the approximated Gamma random variable for the accumulated interference power $\mathsf I$. In this case, there are, on average, 2.0837 cooperative satellites within the cluster area, with the number of satellites within the observable spherical dome $\lambda_{\mathrm S}|\mathcal{A}|$ set to $50$. The second scenario, considering 64,100 satellites, utilizes the distribution of $\tilde{\mathsf{D}}$, the approximated Gamma random variable for the accumulated cluster power $\mathsf D$. This scenario involves an average of 12.5020 cooperative satellites in the cluster area at $\lambda_{\mathrm S}|\mathcal{A}| = 300$.

%\textcolor{red}{It is worth noting that the coverage probabilities with SIR and SIR are nearly identical for scenarios with path loss $\alpha=~2.3$ and satellite counts in the observable spherical dome, In particular $\lambda_{\mathrm S}|\mathcal{A}| = 50$ and $\lambda_{\mathrm S}|\mathcal{A}| = 300$. This similarity results from the interference dominance in densely deployed satellite systems, while noise remains negligible for small $\alpha$.}

\subsection{Validity of Gamma Approximation}
%We extend the prior results from terrestrial networks to satellite networks, transitioning from infinite space to a finite sphere. In this context, we investigate both the Gamma approximation over the accumulated interference power $\mathsf I$ and the accumulated cluster power $\mathsf D$.
Figs.~\ref{fig_gamma_int} and \ref{fig_gamma_clu} show a comparison between the actual CDFs obtained from Monte Carlo simulations and our analytical CDFs. As can be seen in these figures, the Gamma random variables, {$\tilde{\mathsf{I}}$ and $\tilde{\mathsf{D}}$}, closely match the statistics of both \black{the accumulated interference power $\mathsf I$ and the accumulated cluster power $\mathsf D$}, which are defined on the state space of Poisson processes. Although the difference in CDFs is marginal, {the CDFs of $\tilde{\mathsf{D}}$ are slightly more deviated from the CDFs of $\mathsf D$, as shown in Fig.~\ref{fig_gamma_clu}, than the CDFs of $\tilde{\mathsf{I}}$ from the CDFs of $\mathsf I$, as shown in Fig.~\ref{fig_gamma_int}. This is due to the small cluster area, Gamma approximation over $\mathsf{D}$ involves fewer terms than $\mathsf{I}$ given that there are fewer satellites in the cluster area compared to those outside.}
%\green{However, there is a marginal difference in CDFs, which can influence the gap between the actual coverage probability and its lower and upper bounds, in conjunction with the manipulation of the shape parameter. Nevertheless, the validation of the coverage probability bounds, as will be shown in Figs.~\ref{fig_cov_int} and \ref{fig_cov_clu}, substantiates the use of Gamma approximation for modeling $\mathsf I$ and $\mathsf D$ in satellite networks.}

%Although the difference in CDFs is marginal, as will be shown in Figs.~\ref{fig_cov_int} and \ref{fig_cov_clu}, gaps persist between the actual coverage probability and its lower and upper bounds. The deviation from actual results in CDFs can influence these gaps, in conjunction with the manipulation of the shape parameter. 

%\green{It is notable that, due to the small cluster area, Gamma approximation over $\mathsf{D}$ involves fewer terms than $\mathsf{I}$ because there are fewer satellites in the cluster area compared to those outside. Despite this, in Fig.~\ref{fig_gamma_clu}, $\tilde{\mathsf{D}}$ exhibits only a small deviation from the actual results, comparable to $\tilde{\mathsf{I}}$ in Fig.~\ref{fig_gamma_int}.
\begin{figure} [t!]
	\centering
	\includegraphics[width=1\columnwidth]{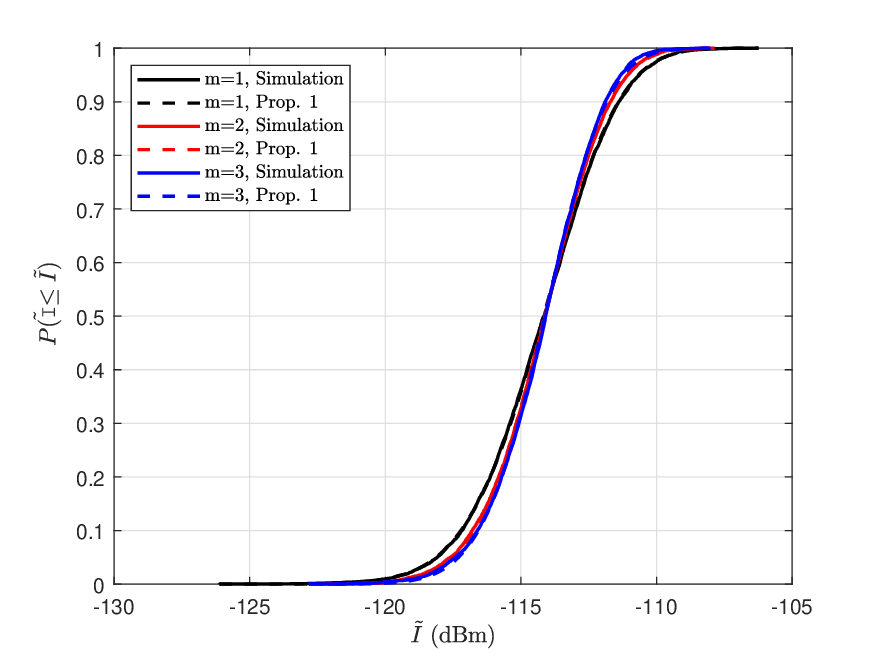}
	\caption{Actual CDFs (solid) and CDFs of the Gamma random variable $\tilde{\mathsf I}$ from Proposition \ref{shape_and_scale} (dashed) ($\lambda_{\mathrm S}|\mathcal{A}| = 50$, scenario~1).}\label{fig_gamma_int}
\end{figure}
\begin{figure}
	\centering
	\includegraphics[width=1\columnwidth]{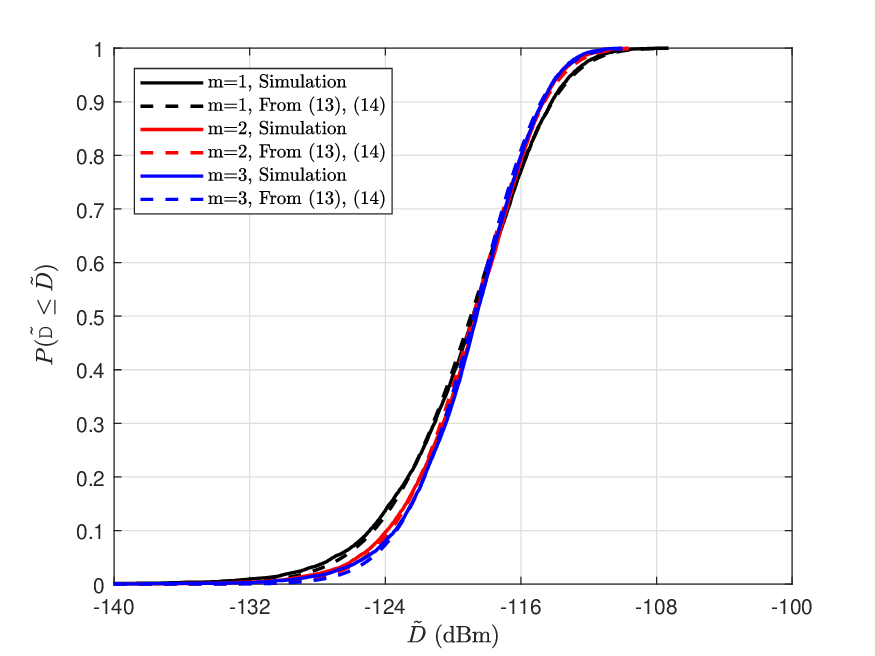}
	\caption{Actual CDFs (solid) and CDFs of the Gamma random variable $\tilde{\mathsf D}$ with parameters in \eqref{shape_D} and \eqref{scale_D} (dashed) ($\lambda_{\mathrm S}|\mathcal{A}| = 300$, scenario~2).}\label{fig_gamma_clu}
\end{figure}

\begin{figure} [t!]
	\centering
	\includegraphics[width=1\columnwidth]{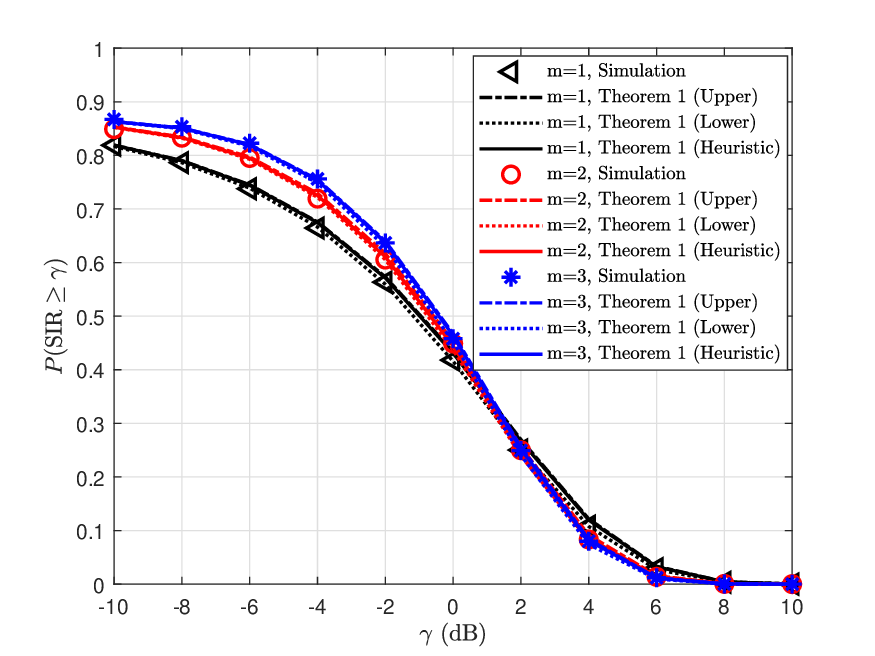}
	\caption{Coverage probabilities from simulation and Theorem~\ref{cov} with heuristic approach \eqref{heu} ($\lambda_{\mathrm S}|\mathcal{A}| = 50$, scenario~1).}\label{fig_cov_int}
\end{figure}

\begin{figure}
	\centering
	\includegraphics[width=1\columnwidth]{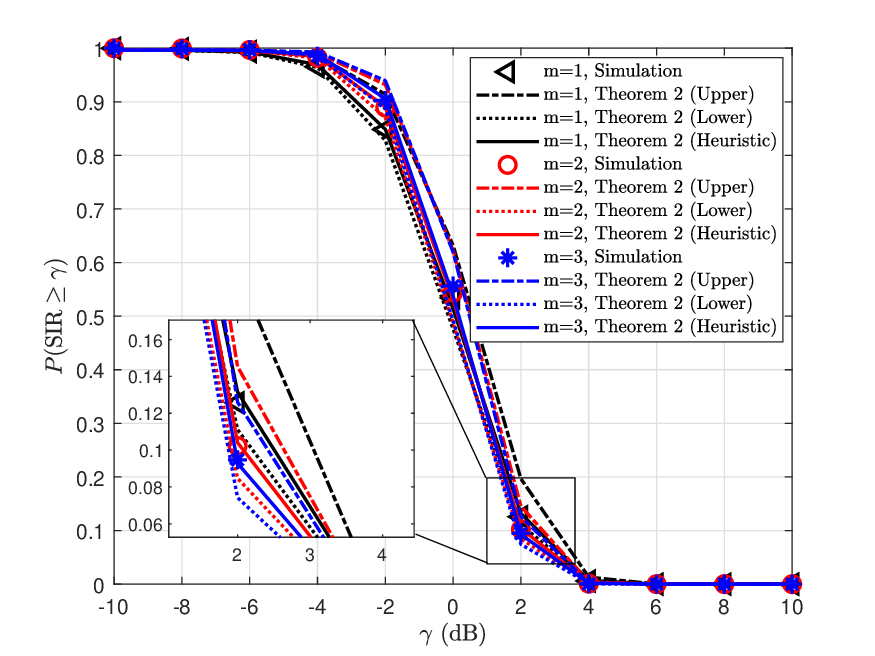}
	\caption{Coverage probabilities from simulation and Theorem~\ref{cov_clu_theorem} with heuristic approach \eqref{heu} ($\lambda_{\mathrm S}|\mathcal{A}| = 300$, scenario~2).}\label{fig_cov_clu}
\end{figure}
\subsection{Derived Coverage Probability} \label{sim_sector}
In Figs.~\ref{fig_cov_int} and \ref{fig_cov_clu}, we present the lower and upper bounds of the coverage probability based on $\tilde{\mathsf{I}}$ and $\tilde{\mathsf{D}}$, in scenarios~1 and~2, respectively. It can be observed that both bounds effectively encompass the actual coverage probability, as obtained from Monte Carlo simulations, in both scenarios. Additionally, we suggest a heuristic approach, denoted as $P^{\mathrm{cov,heu}}$, to effectively emulate the coverage probability using the derived bounds. \black{This heuristic approach bridges the gap between the lower and upper bounds by linearly integrating them with weights that reflect the relative distance of the shape parameter~$k$ from its nearest integers, $\lfloor k \rfloor$ and $\lceil k \rceil$. To consider the impact of rounding on $k$ for the coverage probability bounds, we use the terms $(\lceil k \rceil -k)$ and $(k-\lfloor k \rfloor)$, each normalized by $\lceil k \rceil - \lfloor k \rfloor=1$, as multipliers for the respective bounds. Specifically, in the context of Theorem \ref{cov}, since the upper bound is associated with $\lfloor k \rfloor$ and the lower bound with $\lceil k \rceil$, these multipliers indicate {that the closer the rounded integer is to~$k$, the greater the weight assigned by $P^{\mathrm{cov,heu}}$ to the corresponding bound.} The expression for the proposed heuristic approach is formulated as}
%\blue{This methodological refinement enhances the precision for our coverage probability estimations, ensuring that they closely mirror actual system performance.}
\begin{figure} [t!]
	\centering
	\includegraphics[width=1\columnwidth]{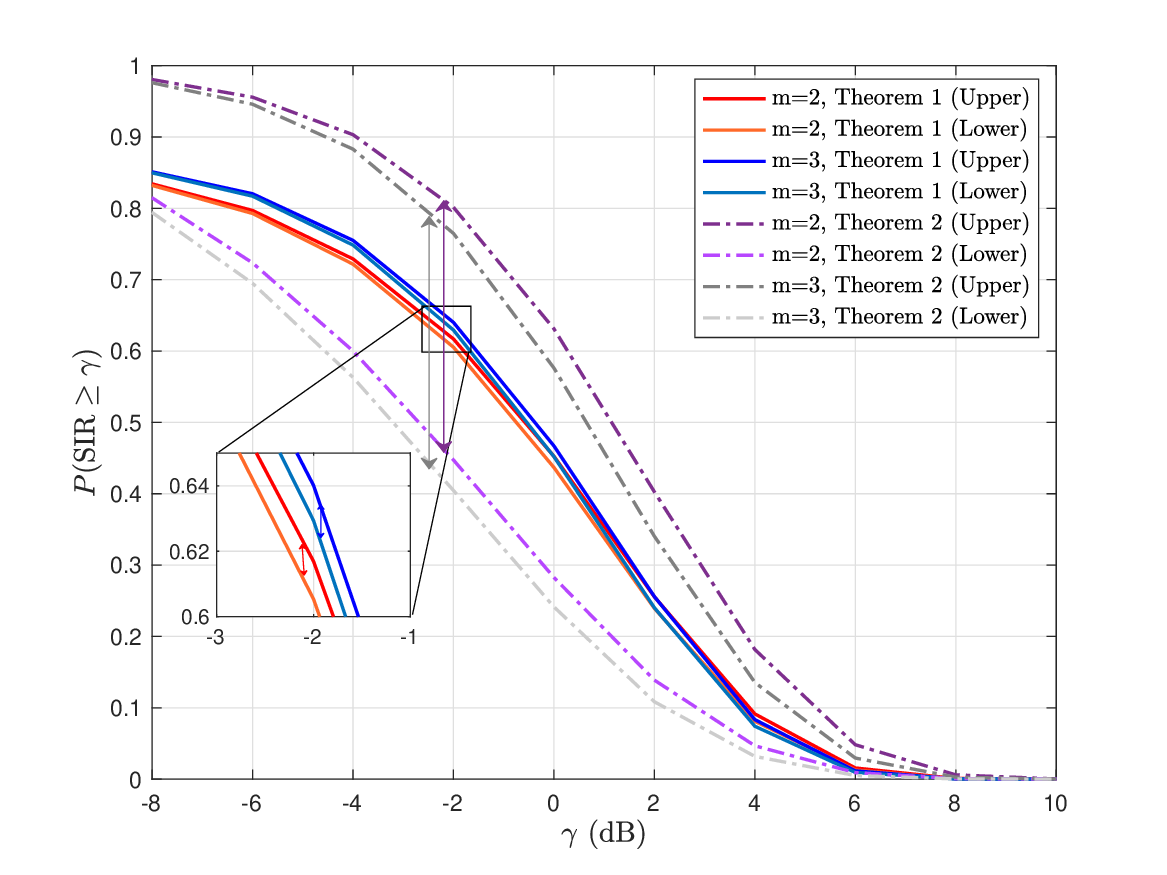}
	\caption{Comparing the tightness of the coverage probability bounds from Theorems \ref{cov} and \ref{cov_clu_theorem} ($\lambda_{\mathrm S}|\mathcal{A}| = 50$, scenario 1). \black{The shape parameters for $\tilde{\mathsf I}$ and $\tilde{\mathsf D}$ are set to $26.4586$ and $1.56$ when $m=2$, and $29.7660$ and $1.04$ when $m=3$, respectively.}}\label{fig_cov_comparsion}
\end{figure}

%\begingroup
%\allowdisplaybreaks
\begin{align} \label{heu}
P^{\mathrm{cov,heu}} \notag \\
&\hspace{-3.25 pc}={(\lceil k \rceil -k )}\underbrace{\Bigg(1-\sum_{n=0}^{\lfloor{k}\rfloor-1}\frac{(\gamma\theta_{\tilde{\mathsf I}})^{-n}}{n!}(-1)^n \frac{d^n \mathcal{L}_{\{\mathsf D \}}(s)}{ds^n}\Bigg|_{s=\frac{1}{\gamma \theta_{\tilde{\mathsf I}}}}\Bigg)}_{\black{\text{Upper bound of $P^{\mathrm{cov}}$}}}\notag \\
&\hspace{-3.25 pc}+{(k-\lfloor k \rfloor)} \underbrace{\Bigg(1-\sum_{n=0}^{\lceil{k}\rceil-1}\frac{(\gamma\theta_{\tilde{\mathsf I}})^{-n}}{n!}(-1)^n  \frac{d^n \mathcal{L}_{\{\mathsf D \}}(s)}{ds^n}\Bigg|_{s=\frac{1}{\gamma \theta_{\tilde{\mathsf I}}}}\Bigg)}_{\black{\text{Lower bound of $P^{\mathrm{cov}}$}}}.
\end{align}
%\endgroup
We can easily obtain the heuristic expression for the bounds obtained in Theorem \ref{cov_clu_theorem} using the same principle. \black{Given that the shape parameter $k$ is influenced by factors such as the number of satellites and the cluster area, effectively reflecting the impact of rounding operations on $k$ serves as a crucial bridge between theoretical upper and lower bounds and Monte Carlo-based system performance, as demonstrated in Figs.~\ref{fig_cov_int} and \ref{fig_cov_clu}.}
%As the heuristic approach includes the effect of round processes by considering the weight, a close match with the actual result can be observed in both Figs.~\ref{fig_cov_int} and \ref{fig_cov_clu}.

\black{In Fig.~\ref{fig_cov_comparsion}, we compare the tightness of the lower and upper bounds of the coverage probability. The bounds derived from~$\tilde{\mathsf{I}}$ show greater tightness than those from $\tilde{\mathsf{D}}$. This difference in tightness results from the larger shape parameter of $\tilde{\mathsf{I}}$ than that of $\tilde{\mathsf{D}}$. As detailed in \eqref{cov_bound} and \eqref{cov_bound_2}, the shape parameter, {$k_{\tilde{\mathsf{I}}}$ or $k_{\tilde{\mathsf{D}}}$}, influences the number of operations required, with the upper bound incorporating one extra operation compared to the lower bound. A higher number of operations typically results in reduced incremental differences between successive calculations, leading to a narrower gap between the lower and upper bounds. Therefore, the coverage probability bounds from~$\tilde{\mathsf{I}}$, with its high shape parameter, are more refined than those from~$\tilde{\mathsf{D}}$, although both adequately encompass Monte Carlo simulations.} However, there is a trade-off between the tightness and computational complexity, as a large shape parameter introduces an increased computational burden. In scenarios involving a massive number of satellites, the coverage probability bounds derived from $\tilde{\mathsf{D}}$ can provide a fine bound while demanding significantly fewer computational resources. For instance, considering an average of $300$ satellites in the observable spherical dome $\mathcal{A}$ with $m=2$, the shape parameters for $\tilde{\mathsf{I}}$ and $\tilde{\mathsf{D}}$ are $k_{\tilde{\mathsf{I}}}=158.7518$ and $k_{\tilde{\mathsf{D}}}=8.3198$, respectively. In this case, the coverage probability bounds utilizing $\tilde{\mathsf{D}}$ can be computed within a manageable time frame, thanks to the smaller shape parameter. {It is worth noting that the tightness of the bounds derived from $\tilde{\mathsf{D}}$ improves as {the system includes more satellites or the cluster area expands}.} \black{By proposing two distinct methods based on Theorems~1 and~2, we facilitate a tailored selection for analyzing coverage probability specific to the system model under consideration. This flexibility enables researchers to adapt their strategies according to the system's demands, optimizing either for tightness or computational efficiency.}

\subsection{Impacts of {System Parameters}}
We {first} investigate the impact of the cluster area, \black{determined by the polar angle $\phi_{\mathrm{clu}}$, which represents the maximum zenith angle of the multiple satellites in the cluster area $\mathcal{A}_{\mathrm{clu}}$,} on the coverage probability. As $\phi_{\mathrm{clu}}$ increases, the cluster area expands, incorporating satellites that previously caused interference into the cluster area for cooperative service to the user. Hence, the coverage probability is enhanced due to the increased accumulated cluster power and decreased accumulated interference power. {Fig.~\ref{fig_ang_clu}} shows the coverage probabilities for various SIR thresholds and $\phi_{\mathrm{clu}}$. As $\phi_{\mathrm{clu}}$ increases, the slope initially rises and then decreases beyond a specific $\phi_{\mathrm{clu}}$ for each SIR threshold. For example, with a SIR threshold $\gamma = -5$ dB, the coverage probability remarkably increases as $\phi_{\mathrm{clu}}$ rises from $1^\circ$ to $2^\circ$ but shows a slight increase from $2^\circ$ to $3^\circ$. Additionally, for a given $\phi_{\mathrm{clu}}$, we can determine the effective SIR threshold for a specific demanded coverage probability. For instance, if the cluster area is physically constrained to $\phi_{\mathrm{clu}}= 3^\circ$, setting the threshold to $0$ dB rather than $-10$ dB is more efficient to enhance the network performance.

%Similarly, Fig.~\ref{fig_ang_clu_lin} shows the impact of the cluster area on coverage probabilities using the linear antenna gain model with parameters $a=1$ and $b=0.1$. These coverage probabilities consistently exhibit lower values than those with the sectored antenna gain for the same SIR threshold. This is because the linear antenna gain model considers the continuity, while the sectored antenna gain is modeled in a discrete form. In particular, in the linear antenna gain model, the antenna gain of the satellite linearly decreases from 1 to 0.1 as its distance from the user ranges from $R_{\mathrm{min}}$ to $R_{\mathrm{max}}$. In contrast, for a sectored antenna gain model, the antenna gain is 1 or 0.1 depending on whether the satellite is within the cluster area or outside. Consequently, in the linear model, most satellites within the cluster area have lower antenna gain than the sectored model, and most satellites outside the cluster area have higher gain than the sectored model. The observed difference in coverage probabilities between Figs.~\ref{fig_ang_clu} and \ref{fig_ang_clu_lin} suggests that a more refined antenna gain model may require the revision of network parameters, such as the cluster area and SIR threshold.
\begin{figure} [t!]
	\centering
	\includegraphics[width=1\columnwidth]{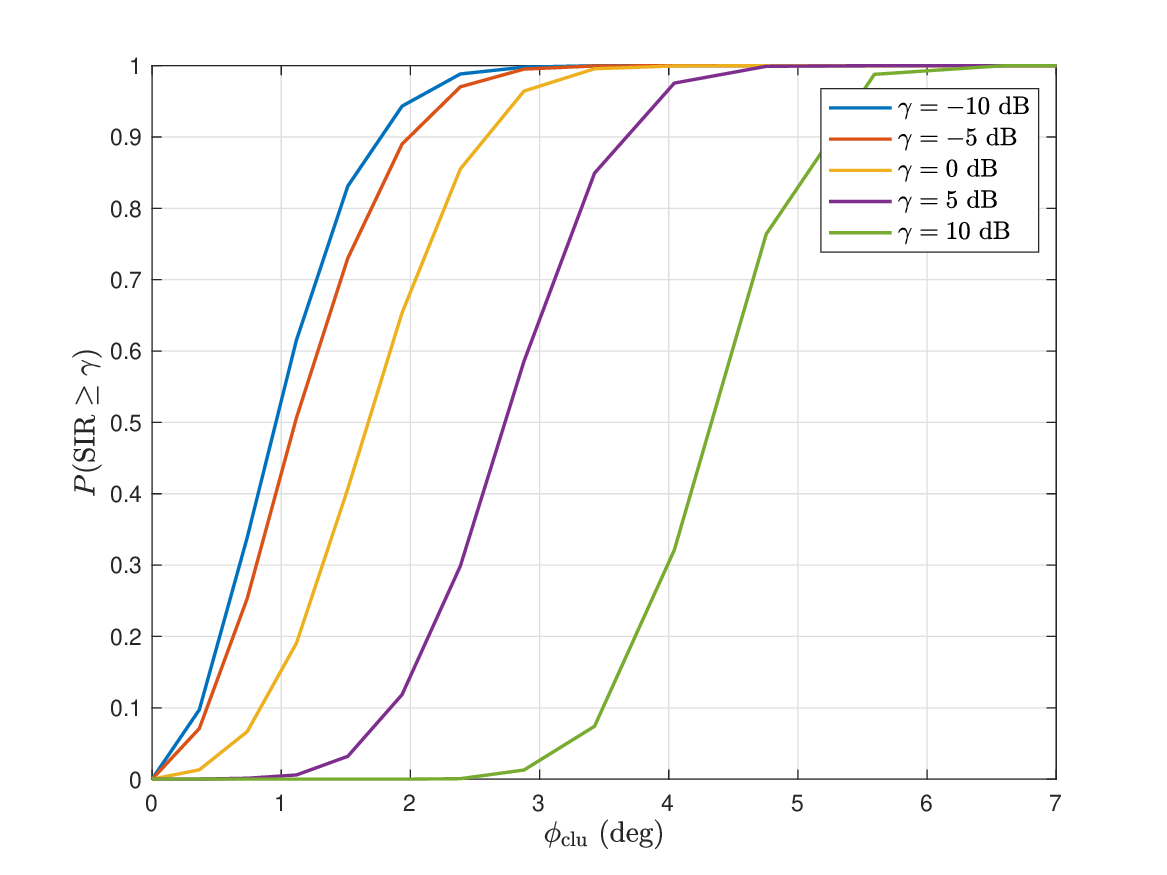}
	\caption{Coverage probabilities corresponding to $\gamma$ with respect to $\phi_{\mathrm {clu}}$ when $m=2$ ($\lambda_{\mathrm S}|\mathcal{A}| = 50$, scenario 1).}\label{fig_ang_clu}
\end{figure}

\begin{figure}
	\centering
	\includegraphics[width=1\columnwidth]{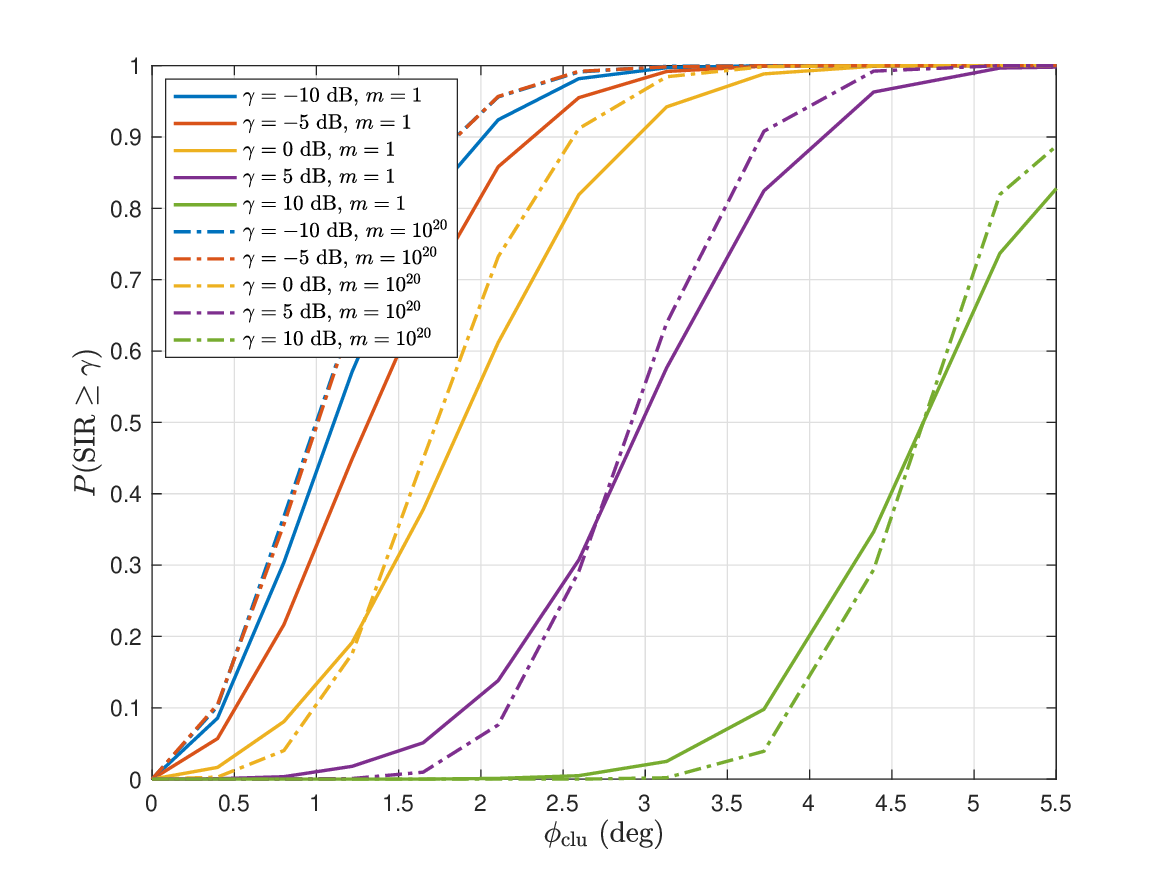}
	\caption{\black{Coverage probabilities corresponding to $m=1,10^{20}$ and $\gamma$ with respect to $\phi_{\mathrm {clu}}$ ($\lambda_{\mathrm S}|\mathcal{A}| = 50$, scenario 1).}}\label{fig_nakagami}
\end{figure}

In Fig.~\ref{fig_nakagami}, we delve into the relationship between coverage probability and the Nakagami parameter $m$, which is linked to the Rician $K$ factor by the formula $m=(K+1)^2/(2K+1)$, quantifying the intensity of the LOS effect. As $m$ increases, the LOS effect is intensified due to $K=(m-1)+\sqrt{m^2-m}$ and $\frac{dK}{dm}=1+\frac{2m-1}{2\sqrt{m^2-m}}>0$. This intensified LOS effect increases both the accumulated cluster power and interference power, which can be explained by considering the derivatives of shape and scale parameters of their approximated Gamma random variables, $\tilde{\mathsf{D}}$ and $\tilde{\mathsf{I}}$, with respect to $m$.
The derivatives of the shape and scale parameters of~$\tilde{\mathsf{D}}$ and~$\tilde{\mathsf{I}}$ are detailed below
\begingroup
\allowdisplaybreaks
\begin{align} \label{m_shape_D}
\frac{dk_{\tilde{\mathsf D}}}{dm}=\frac{4(\alpha-1)\pi \lambda_{\mathrm S} \frac{R_{\mathrm {S}}}{R_{\mathrm {E}}}}{(\alpha-2)^2 }\frac{(R_{\mathrm{min}}^{-\alpha+2}-R_{\mathrm{clu}}^{-\alpha+2})^2}{R_{\mathrm{min}}^{-2\alpha+2}-R_{\mathrm{clu}}^{-2\alpha+2}}\frac{1}{m^2\left(1+\frac{1}{m}\right)^2}, 
\end{align}
\begin{align} \label{m_scale_D}
\frac{d\theta_{\tilde{\mathsf D}}}{dm}=\frac{(\alpha-2) G^{\mathrm t}_{\mathrm i} }{2(\alpha-1)}\frac{R_{\mathrm{min}}^{-2\alpha+2}-R_{\mathrm{clu}}^{-2\alpha+2}}{R_{\mathrm{min}}^{-\alpha+2}-R_{\mathrm{clu}}^{-\alpha+2}}\left(-\frac{1}{m^2}\right),
\end{align}
\begin{align} \label{m_shape_I}
\frac{dk_{\tilde{\mathsf I}}}{dm}=\frac{4(\alpha-1)\pi \lambda_{\mathrm S} \frac{R_{\mathrm {S}}}{R_{\mathrm {E}}}}{(\alpha-2)^2}\frac{(R_{\mathrm{clu}}^{-\alpha+2}-R_{\mathrm{max}}^{-\alpha+2})^2}{R_{\mathrm{clu}}^{-2\alpha+2}-R_{\mathrm{max}}^{-2\alpha+2}}\frac{1}{m^2\left(1+\frac{1}{m}\right)^2},
\end{align}
\begin{align} \label{m_scale_I}
\frac{d\theta_{\tilde{\mathsf I}}}{dm}=\frac{(\alpha-2)G^{\mathrm t}_{\mathrm o} }{2(\alpha-1)} \frac{R_{\mathrm{clu}}^{-2\alpha+2}-R_{\mathrm{max}}^{-2\alpha+2}}{R_{\mathrm{clu}}^{-\alpha+2}-R_{\mathrm{max}}^{-\alpha+2}}\left(-\frac{1}{m^2}\right).
\end{align}
\endgroup
It can be easily verified that the signs of derivatives \eqref{m_shape_D} and \eqref{m_shape_I} are positive, indicating that increasing $m$ raises the shape parameters of $\tilde{\mathsf{D}}$ and $\tilde{\mathsf{I}}$, and the signs of derivatives \eqref{m_scale_D} and \eqref{m_scale_I} are negative, indicating that an increase in $m$ reduces the scale parameters of $\tilde{\mathsf{D}}$ and~$\tilde{\mathsf{I}}$. Since a higher shape parameter increases the average value of the Gamma random variable, and a smaller scale parameter reduces variability around the mean, the results in \eqref{m_shape_D}-\eqref{m_scale_I} confirm that as $m$ increases -- intensifying the LOS effect -- both $\tilde{\mathsf{D}}$ and $\tilde{\mathsf{I}}$ increase. {Moreover, as detailed in \eqref{m_shape_D} and \eqref{m_shape_I}, the \black{increases} in $\tilde{\mathsf{D}}$ and $\tilde{\mathsf{I}}$ approach zero as $R_{\mathrm{clu}}$ is close to $R_{\mathrm{min}}$ and $R_{\mathrm{max}}$, respectively.} Therefore, for large $R_{\mathrm{clu}}$ (high $\phi_{\mathrm{clu}}$), coverage probability increases with rising~$m$ due to a relatively higher increase in accumulated cluster power compared to the increase in accumulated interference power. Conversely, coverage probability decreases with an increase in $m$ in scenarios with small $R_{\mathrm{clu}}$ (low $\phi_{\mathrm{clu}}$), as the predominant increase in accumulated interference power exceeds the rise in accumulated cluster power. In this context, at high SIR thresholds, a sufficiently high $\phi_{\mathrm{clu}}$ is necessary because an increase in accumulated cluster power is essential for improving coverage while accumulated interference power concurrently rises with increasing $m$. In other words, as SIR thresholds decrease, the crossing from a coverage decrease to increase with rising $m$ occurs at lower~$\phi_{\mathrm{clu}}$ levels, as observed in the leftward shift of the coverage crossing point with decreasing SIR thresholds. \black{This analysis concludes that increasing $m$ does not always improve coverage in satellite networks using the cluster area-based approach. It underscores the necessity of incorporating other system parameters to accurately evaluate the impact of satellite clustering and strategically design satellite networks.}

\section{Conclusions}\label{sec5}
In this paper, we proposed mathematical analyses of the performance of satellite cluster networks, where satellites in the cluster area collaborate to serve users, particularly in mega-constellations. We demonstrated that satellite clustering significantly enhances coverage probability compared to solely utilizing the nearest satellite. To establish the lower and upper bounds of the coverage probability mathematically, we first derived the key parameters of the approximated Gamma random variables to address the compound terms in the SIR based on the system model. Leveraging the distribution of two distinct approximated Gamma random variables, we suggested the lower and upper bounds of the coverage probability in two different approaches and compared the advantages of each in terms of tightness and complexity. These bounds that depend on system parameters effectively showed the network's performance with sufficient tightness to simulation results. Moreover, regarding satellite cluster networks, our mathematical analyses of how system parameters such as cluster area, LOS intensity, and SIR thresholds impact coverage probability can provide efficient strategies for designing dense satellite networks that necessitate effective interference management.

%Moreover, our analyses not only reduced computational burden through finite operations but also expected to offer insights into the impact of system parameters on the coverage probability in satellite cluster networks.

\begin{appendices}
\section{Proof of Proposition~\ref{shape_and_scale}}\label{appendix_shape_and_scale}
For a Gamma random variable $\mathsf X \sim \Gamma(k,\theta)$, its shape parameter $k$ and scale parameter $\theta$ are determined by the first and second-order moments $k\theta=\mathbb{E}[\mathsf X]$ and $k\theta^2=\text{Var}[\mathsf X]$.
By the Campbell's theorem \cite{Kingman:1992},
the mean and the variance of the approximated Gamma random variable $\tilde{\mathsf I}$ for the accumulated interference power ${\mathsf I}$ are derived as follows

\begin{align}\label{shape}
\mathbb{E}(\tilde{\mathsf I})&\stackrel{(a)}{=}\black{\int_{ \text{Area of} \bar{\mathcal{A}}_{\mathrm{clu}} } G^{\mathrm t}_{\mathrm o} \mathbb{E}[{\mathsf H}] r^{-\alpha} \lambda_{\mathrm S} dr } \notag\\
&\stackrel{(b)}{=}\black{2 \frac{R_{\mathrm {S}}}{R_{\mathrm {E}}}\pi \lambda_{\mathrm S} \int_{R_{\mathrm{clu}}}^{R_{\mathrm{max}}} G^{\mathrm t}_{\mathrm o} \mathbb{E}[{\mathsf H}] r^{-\alpha}r dr} \notag\\
%&=\frac{2 G^{\mathrm t}_{\mathrm o}}{\alpha-2}\frac{R_{\mathrm {S}}}{R_{\mathrm {E}}} \pi\lambda_{\mathrm S} \mathbb{E}[{\mathsf H}](R_{\mathrm{clu}}^{-\alpha+2}-R_{\mathrm{max}}^{-\alpha+2})\notag \\
&=\frac{2 G^{\mathrm t}_{\mathrm o}}{\alpha-2}\frac{R_{\mathrm {S}}}{R_{\mathrm {E}}} \pi  \lambda_{\mathrm S} (R_{\mathrm{clu}}^{-\alpha+2}-R_{\mathrm{max}}^{-\alpha+2}), 
\end{align}
\begin{align} \label{scale}
\mathrm{Var}(\tilde{\mathsf I})&\stackrel{(a)}{=}\black{\int_{\text{Area of } \bar{\mathcal{A}}_{\mathrm{clu}} } ({G^{\mathrm t}_{\mathrm o}})^2 \mathbb{E}[{\mathsf H}^2] r^{-2\alpha} \lambda_{\mathrm S} dr}\notag\\
&\stackrel{(b)}{=}\black{2 \frac{R_{\mathrm {S}}}{R_{\mathrm {E}}}\pi \lambda_{\mathrm S} \int_{R_{\mathrm{clu}}}^{R_{\mathrm{max}}} ({G^{\mathrm t}_{\mathrm o}})^2 \mathbb{E}[{\mathsf H}^2] r^{-2\alpha}r dr} \notag\\
%&\hspace{0 pc}=\frac{2({G^{\mathrm t}_{\mathrm o}})^2}{2\alpha-2}\frac{R_{\mathrm {S}}}{R_{\mathrm {E}}} \pi \lambda_{\mathrm S} \mathbb{E}[{\mathsf H}^2](R_{\mathrm{clu}}^{-2\alpha+2}-R_{\mathrm{max}}^{-2\alpha+2})\notag \\
&\hspace{-0 pc}=\frac{2({G^{\mathrm t}_{\mathrm o}})^2}{2\alpha-2}\frac{R_{\mathrm {S}}}{R_{\mathrm {E}}} \pi \lambda_{\mathrm S} \left(1+\frac{1}{m}\right)(R_{\mathrm{clu}}^{-2\alpha+2}-R_{\mathrm{max}}^{-2\alpha+2}),
\end{align}
\black{where both {(a)'s} in \eqref{shape} and \eqref{scale} are derived from calculating expectations of sums of measurable functions over the PPP with an intensity measure characterized by a density~$\lambda_\mathrm S$ within the surface area $\bar{\mathcal{A}}_{\mathrm{clu}}$, and both {(b)'s} follow from the surface area of a spherical dome, defined as $|\mathcal{A}_{r}|=2\pi R_\mathrm  S \left(R_\mathrm  S-R_\mathrm  E -\frac{R_\mathrm  S^2-R_\mathrm  E^2 -r^2}{2 R_\mathrm  E}\right)$, with $\frac{\partial |\mathcal{A}_r|}{\partial r}=2\frac{R_\mathrm  S}{R_\mathrm  E}\pi r$ accordingly.}

By substituting \eqref{shape} and \eqref{scale} into the relations of the Gamma random variable, $k=(\mathbb{E}[\mathsf X])^2 /\text{Var}[\mathsf X]$ and $\theta=~\text{Var}[\mathsf X] / \mathbb{E}[\mathsf X]$, $k_{\tilde{\mathsf I}}$ and $\theta_{\tilde{\mathsf I}}$ are given by \eqref{shape_I} and \eqref{scale_I}, which completes the proof.

%The same procedure can be applied to the accumulated cluster power $\mathsf D$ by modifying the area to range from $R_{\mathrm{min}}$ to $R_{\mathrm{clu}}$ instead of $R_{\mathrm{clu}}$ to $R_{\mathrm{max}}$.

\section{Proof of Lemma~\ref{lap_clu}}\label{appendix_lap}
To derive the Laplace transform of the accumulated cluster power $\mathsf D$, we first derive the Laplace transform conditioned on $\mathsf{N}_{\mathrm{S,clu}}=L$, representing the condition that $L$ satellites are in the cluster area $\mathcal{A}_{\mathrm{clu}}$. The conditional Laplace transform of $\mathsf D$ is given by 
\begin{align} \label{cond_lap}
\mathcal{L}_{\{\mathsf D|\mathsf{N}_{\mathrm{S,clu}}=L\}}(s) \notag \\
&\hspace{-5.9 pc}\stackrel{(a)}{=}\left(\mathbb{E}\left[\exp\left({-s G^{\mathrm t}_{\mathrm i}  {\mathsf H}_l||\mathbf{x}_l-\mathbf{u}_1||^{-\alpha}}\right)\right]\right)^L\notag \\ 
&\hspace{-5.9 pc}\stackrel{(b)}{=}\left(\int_{R_{\mathrm{min}}}^{R_{\mathrm{clu}}} \frac{r}{R_{\mathrm {E}} R_\mathrm  S (1-\cos\phi_{\mathrm{clu}})}\mathbb{E}\left[ \exp\left({-s G^{\mathrm t}_{\mathrm i} {\mathsf H}_l r^{-\alpha}}\right)\right] dr\right)^L \notag \\ 
&\hspace{-5.9 pc}\stackrel{(c)}{=}\Bigg(\frac{1}{R_{\mathrm {E}}R_\mathrm  S (1-\cos\phi_{\mathrm{clu}})\alpha} \int_{0}^{\infty}\int_{R_{\mathrm{clu}}^{-\alpha}}^{R_{\mathrm{min}}^{-\alpha}} t ^{-\frac{2}{\alpha}-1} \notag \\
&\hspace{3.5 pc}\cdot e^{-s G^{\mathrm t}_{\mathrm i} ht}\frac{m^m e^{-mh}h^{m-1}}{\Gamma(m)}dtdh\Bigg)^L \notag \\
&\hspace{-5.9 pc}\stackrel{(d)}{=}\Bigg(\frac{1}{R_{\mathrm {E}}R_\mathrm  S (1-\cos\phi_{\mathrm{clu}})\alpha} \int_{R_{\mathrm{clu}}^{-\alpha}}^{R_{\mathrm{min}}^{-\alpha}}t ^{-\frac{2}{\alpha}-1} \int_{0}^{\infty} e^{-v}v^{m-1} \notag \\
&\hspace{-1.5 pc} \cdot (s G^{\mathrm t}_{\mathrm i} t+m)^{-m+1} \frac{m^m}{\Gamma(m)} (s G^{\mathrm t}_{\mathrm i} t+m)^{-1}dvdt \Bigg)^L,
\end{align}
where (a) is obtained using the definition of $\mathsf D$ and the Laplace transform, (b) follows from the conditional PDF $f_{\mathsf R| \mathbf x \in \mathcal{A}_{\mathrm{clu}}}(r)$ given in Lemma \ref{pdf_dis}, (c) results from the change of variable $r^{-\alpha}=t$ and the PDF for $\mathsf{H}_l$ given in \eqref{nak}, and (d) uses the change of variable $v=h(s G^{\mathrm t}_{\mathrm i} t+m)$. The final expression in \eqref{lap_clu_int} is derived by using the definition of the Gamma function $\int_{0}^{\infty} e^{-v}v^{m-1}dv=\Gamma(m)$, which completes the proof.

Additionally, it is noteworthy that the conditional Laplace transform, considered in deriving the coverage probability, remains independent of the number of satellites $L$. Taking the ${L}$th root of \eqref{lap_clu_int}, it becomes evident that the resulting  ${L}$th root of the conditional Laplace transform of $\mathsf D$ is unaffected by the number of satellites $L$.
%with $dh(s G^{\mathrm t}_{\mathrm i} t+m)=dv$
%$dr=t^{-1-\frac{1}{\alpha}}\left(-\frac{1}{\alpha}\right)dt$

\section{Proof of Theorem \ref{cov}} \label{appendix_cov}
The Erlang distribution is a special case of the Gamma distribution with a positive integer shape parameter. In the case of the Gamma random variable $\mathsf X$ with the shape parameter $k$, its lower and upper CDF bounds are represented using the Erlang distribution with the shape parameters $\lceil k \rceil$ (lower bound) and $\lfloor k \rfloor$ (upper bound) \cite{Tan:2014},\cite{Kara:2006}. The expression is as follows
\begin{equation} \label{cdf}
	\mathbb {P}[\mathsf{X}\leq x] {{\tilde{k}= \lceil k \rceil\atop\geq}\atop{<\atop \tilde{k}=\lfloor k \rfloor}} 1- e^{-x/\theta}\sum _{n=0}^{\tilde{k}-1}\frac {(x/\theta)^{n}}{n!}.
\end{equation}
In the same context, we transform the approximated Gamma distribution of the accumulated interference power $\tilde{\mathsf I} \sim ~\Gamma(k_{\tilde{\mathsf I}},\theta_{\tilde{\mathsf I}})$ to the Erlang distribution by round-up or -down the shape parameter $k_{\tilde{\mathsf I}}$ to an integer $\tilde{k}$.

To derive the lower and upper bounds of the coverage probability, the coverage probability with respect to $\mathsf I$ is revisited
\begin{equation}\label{cov_I}
P^{\mathrm{cov}}(\gamma; \lambda_{\mathrm S}, R_{\mathrm {S}},\phi_{\mathrm{clu}},\alpha,m) =\mathbb{E}\left[P\left(\mathsf I \leq \frac{\mathsf D}{ \gamma} \right)\right].
\end{equation}
Invoking \eqref{cdf} into \eqref{cov_I}, we first derive the upper bound of the coverage probability with $\tilde{k}=\lfloor k \rfloor$ as
\begingroup
\allowdisplaybreaks
\begin{align} \label{cov_upper}
&\hspace{-0.5 pc} P^{\mathrm{cov,upper}}(\gamma; \lambda_{\mathrm S}, R_{\mathrm {S}},\phi_{\mathrm{clu}},\alpha,m) \notag \\
&\stackrel{(a)}{=} \mathbb{E}\left[1-\sum_{n=0}^{\tilde{k}-1}e^{-\frac{\mathsf D}{\gamma \theta_{\tilde{\mathsf I}}}}\frac {(\frac{\mathsf D}{\gamma \theta_{\tilde{\mathsf I}}})^{n}}{n!}\right] \notag \\
&\stackrel{(b)}{=}1-\sum_{n=0}^{\tilde{k}-1}\frac{(\gamma\theta_{\tilde{\mathsf I}})^{-n}}{n!}(-1)^n \frac{d^n \mathcal{L}_{\{\mathsf D \}}(s)}{ds^n}\Bigg|_{s=\frac{1}{\gamma\theta_{\tilde{\mathsf I}}}},
\end{align}
\endgroup
where (a) comes from the Erlang distribution with integer shape parameter $\tilde{k}$, and (b) follows from the first derivative property of the Laplace transform $\mathbb{E}[t^k e^{-st}]=(-1)^k \frac{d^k\mathcal{L}(s)}{ds^k}$.
Replacing $\lfloor k \rfloor$ with $\lceil k \rceil$, the lower bound of the coverage probability can be derived using the same process as in \eqref{cov_upper}.
\end{appendices}

\bibliographystyle{IEEEtran}
\bibliography{reference}

%Miyeon Lee (Graduate Student Member, IEEE) received the B.S (Hons.) in electrical engineering from POSTECH in 2021. She is currently pursing Ph.D degree with the School of Electrical Engineering, KAIST. She was a co-recipient of the 2024 ICTC Best SCSS Paper Award. Her research interests are design and analysis of the satellite communication systems and integrated sensing and communication.

  \bstctlcite{IEEEexample:BSTcontrol}
   %\bibliographystyle{IEEEtran}
   %\bibliography{bit_allocation_references}
   
   \begin{IEEEbiography}[{\includegraphics[width=1in,height=1.25in,clip,keepaspectratio]{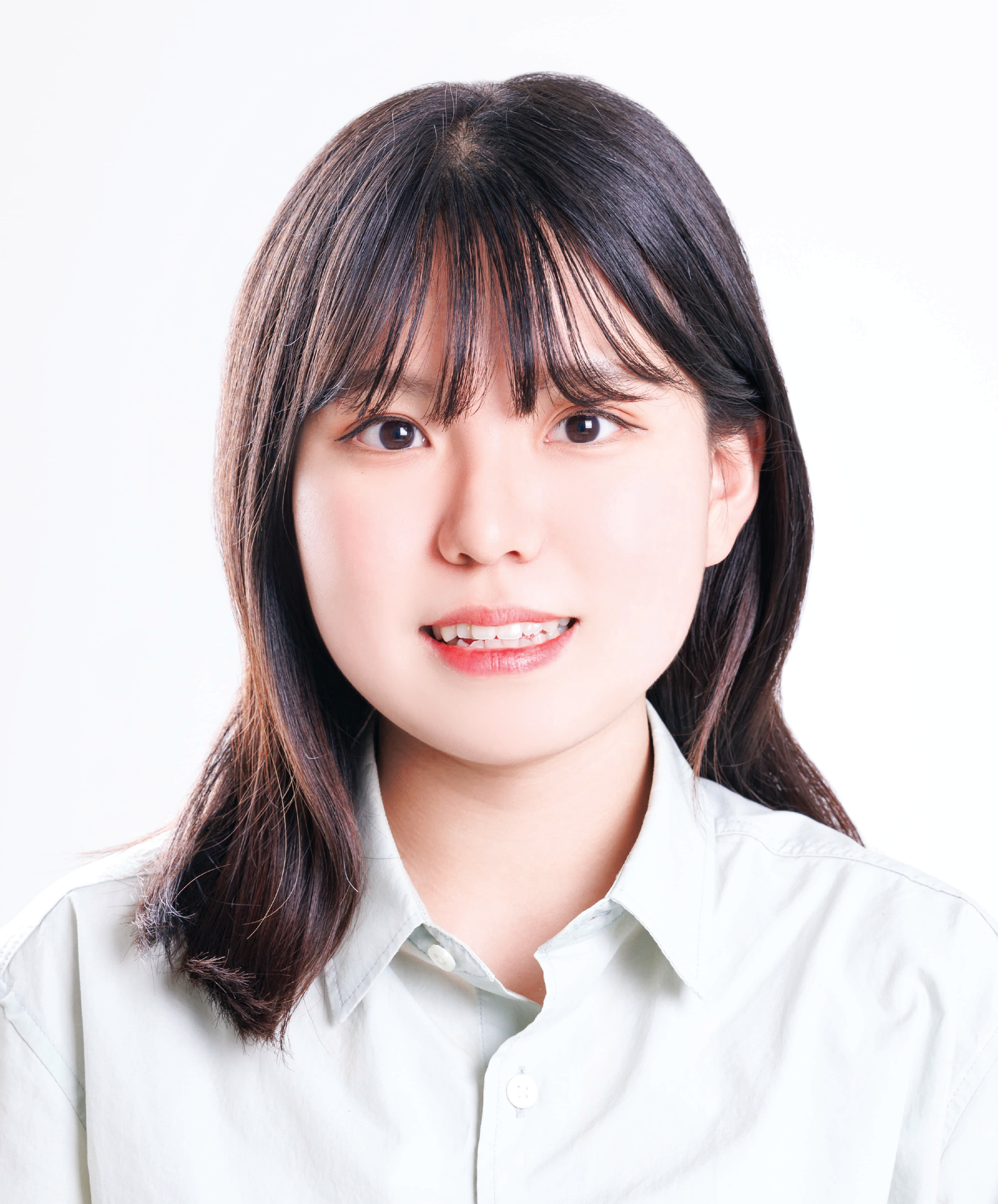}}]{Miyeon Lee}
      (Graduate Student Member, IEEE) received the B.S. (Hons.) degree in Electrical Engineering from Pohang University of Science and Technology (POSTECH) in 2021. She is currently pursuing the integrated M.S. and Ph.D. degree in the School of Electrical Engineering, Korea Advanced Institute of Science and Technology (KAIST). Her research interests include the design and analysis of satellite communication systems and integrated sensing and communication. She was a co-recipient of the 2024 ICTC Best SCSS Paper Award.
   \end{IEEEbiography}
   
    \begin{IEEEbiography}[{\includegraphics[width=1in,height=1.25in,clip,keepaspectratio]{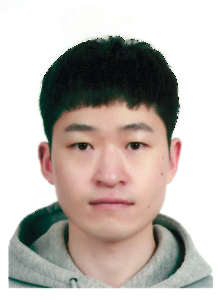}}]{Sucheol Kim}
    (Member, IEEE) received the B.S. and M.S. degrees in electrical engineering from POSTECH in 2017 and 2020, respectively, and the Ph.D. degree in the School of Electrical Engineering from KAIST in 2023. He is working as a Researcher with the Electronics and Telecommunications Research Institute (ETRI). His research interests include system performance analyses, satellite communications, and 3GPP system-level simulations. He was awarded the Outstanding Paper Award of 11th Electronic News ICT Paper Contest in 2019. He was a co-recipient of the 2020 Electronic News ICT Paper Contest Best Paper Award.
   \end{IEEEbiography}

    \begin{IEEEbiography}[{\includegraphics[width=1in,height=1.25in,clip,keepaspectratio]{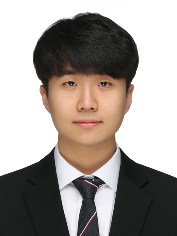}}]{Minje Kim}
      (Graduate Student Member, IEEE) received the B.S. (Hons.) degree in Electrical Engineering from Pohang University of Science and Technology (POSTECH) in 2020. He is currently pursuing the integrated M.S. and Ph.D degree in the School of Electrical Engineering from Korea Advanced Institute of Science and Technology (KAIST). His research interests include the design and analysis of massive MIMO communications and satellite communications. He was a co-recipient of the IITP (Information and Communication Technology Planning and Evaluation Institute) Director’s Award at the ICT Challenge 2021, the Excellence Award at the Electronic Newspaper ICT Paper Contest Grand Exhibition in 2023. He was awarded the KEPCO (Korea Electric Power Corporation) Electrical Engineering Scholarship in 2021.
   \end{IEEEbiography}
   
    \begin{IEEEbiography}[{\includegraphics[width=1in,height=1.25in,clip,keepaspectratio]{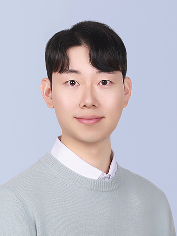}}]{Dong-Hyun Jung}
      (Member, IEEE) received his B.S. degree (with honors) from Pohang University of Science and Technology (POSTECH) in 2015, his M.S. degree from Seoul National University in 2017, and his Ph.D. degree from the Korea Advanced Institute of Science and Technology (KAIST) in 2024. From 2017 to 2025, he was with the Satellite Communications Research Division at the Electronics and Telecommunications Research Institute (ETRI). He also served as an assistant professor at the Department of Information and Communication Engineering, University of Science and Technology (UST), from 2024 to 2025. Since 2025, he has been an assistant professor in the School of Electronic Engineering at Soongsil University. His research interests include non-terrestrial networks, 3GPP system-level simulations, AI-driven communications, and deep space communications.
   \end{IEEEbiography}
   
   \begin{IEEEbiography}[{\includegraphics[width=1in,height=1.25in,clip,keepaspectratio]{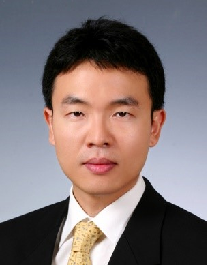}}]{Junil Choi}
      (Senior Member, IEEE) received the B.S. (Hons.) and M.S. degrees in electrical engineering from Seoul National University in 2005 and 2007, respectively, and the Ph.D. degree in electrical and computer engineering from Purdue University in 2015.

He is currently working as a KAIST Endowed Chair Associate Professor with the School of Electrical Engineering, KAIST. From 2007 to 2011, he was a member of technical staff at the Samsung Advanced Institute of Technology (SAIT) and Samsung Electronics Company Ltd., South Korea, where he contributed to advanced codebook and feedback framework designs for the 3GPP LTE/LTE-Advanced and IEEE 802.16m standards. Before joining KAIST, he was a post-doctoral fellow at The University of Texas at Austin from 2015 to 2016 and an assistant professor at POSTECH from 2016 to 2019. His research interests include the design and analysis of massive MIMO, mmWave communications, satellite communications, visible light communications, and communication systems using machine-learning techniques.

Dr. Choi was a co-recipient of the 2022 IEEE Vehicular Technology Society Best Vehicular Electronics Paper Award, the 2021 IEEE Vehicular Technology Society Neal Shepherd Memorial Best Propagation Award, the 2019 IEEE Communications Society Stephen O. Rice Prize, the 2015 IEEE Signal Processing Society Best Paper Award, and the 2013 Global Communications Conference (GLOBECOM) Signal Processing for Communications Symposium Best Paper Award. He was awarded the IEEE ComSoc AP Region Outstanding Young Researcher Award in 2017, the NSF Korea and Elsevier Young Researcher Award in 2018, the KICS Haedong Young Researcher Award in 2019, the IEEE Communications Society Communication Theory Technical Committee Early Achievement Award in 2021, and the 6th Next-Generation Scientist Award the S-OIL Science and Culture Foundation in 2024. He is an IEEE Vehicular Technology Society Distinguished Lecturer, an Area Editor of IEEE Open Journal of the Communications Society and an Associate Editors of IEEE Transactions on Wireless Communications and IEEE Transactions on Communications.

   \end{IEEEbiography}

\end{document}